\documentclass[letterpaper,titlepage]{article}
\usepackage{Amish_Style}

\begin{document}

\begin{center}
    \Large
    \textbf{Stability and Machine Learning Applications of Persistent Homology Using the Delaunay-Rips Complex}

    
    \large
    \vspace{0.4cm}
    Amish Mishra and Francis C. Motta
       
    \vspace{0.9cm}
    \textbf{Abstract}
\end{center}

In this paper we define, implement, and investigate a simplicial complex construction for computing persistent homology of Euclidean point cloud data, which we call the Delaunay-Rips complex (DR). Assigning the Vietoris-Rips weights to simplices, DR experiences speed-up in the persistence calculations by only considering simplices that appear in the Delaunay triangulation of the point cloud. We document and compare a Python implementation of DR with other simplicial complex constructions for generating persistence diagrams. By imposing sufficient conditions on point cloud data, we are able to theoretically justify the stability of the persistence diagrams produced using DR. When the Delaunay triangulation of the point cloud changes under perturbations of the points, we prove that DR-produced persistence diagrams exhibit instability. Since we cannot guarantee that real-world data will satisfy our stability conditions, we demonstrate the practical robustness of DR for persistent homology in comparison with other simplicial complexes in machine learning applications. We find in our experiments that using DR for an ML-TDA pipeline performs comparatively well as using other simplicial complex constructions.


\tableofcontents

\section{Introduction}

As the volume and complexity of data collected in the experimental sciences and in industry applications continues to grow, scientists and engineers often turn to methods that transform data into more compact and manageable representations, while retaining crucial characteristics that enable meaningful data analysis, visualization, and support the development of high-performing predictive models. Tools from the applied computational topology subfield Topological Data Analysis (TDA) offer several solutions to address some of the challenges with processing and analyzing complex or high-dimensional data. TDA leverages results from algebraic topology to measure and quantify qualitative, shape-based features of data and continues to attract the interest of researchers across computational, mathematical, and experimental disciplines including financial networks \cite{leibon2008topological}, signals in images \cite{chung2009persistence}, cancer detection \cite{QAISER2016119}, forensics \cite{10.1007/978-3-319-64185-0_11}, and material science \cite{kramar2014quantifying}.

Persistent homology (PH) is the flagship method within TDA and provides a robust method to encode multi-dimensional geometric/topological features of a dataset in a compact representation known as a persistence diagram (PD). PH is often computed from point cloud data, i.e., a finite collection of points in a metric space. One begins by associating to the cloud a nested family of topological spaces---namely a collection of simplicial complexes consisting of simplices of various dimension (e.g., points, edges, triangles, tetrahedra, etc.)---that is parameterized by a (real) scale parameter. This filtration of simplicial complexes is meant to capture geometric/topological structures in the point cloud across scales. From the filtration, a PD, consisting of a collection of ordered pairs above the diagonal in $\mathbb{R}^2$. Each (birth, death)-pair represents a topological feature of a fixed dimension (connected component, hole, void, etc.) which appears at the birth scale and disappears at the death scale. 

Construction of a filtration depends on two factors: which simplices to include in the family of complexes, and determination of the scales at which each simplex should appear. Both factors contribute to the overall computational burden to compute PH on point cloud data, and may represent barriers to deploying PH on large or high-dimensional data sets. For instance, an appealing and often used construction known as the Rips filtration requires only knowledge of the pairwise dissimilarities between points. This makes computing the scales at which all simplices should appear quite straightforward. However, the size of the Rips complex grows exponentially in the number of points. An alternative to the Rips filtration is the Alpha complex filtration \cite{Edelsbrunner1993TheUO}, which restricts the complexes to simplices in the Delaunay triangulation of the cloud. However, this method assigns scales to simplices according to the neighborship of the Voronoi cells associated with each data point, which can be an expensive computation. Some other solutions \cite{sheehy2012linear, Guibas2007ReconstructionUW, Silva2004TopologicalEU} have more recently been proposed to reduce the size of the complexes ingested by the PH algorithm, which has been implemented in a variety of languages (C++, Python, Julia, etc.) and optimized to further improve computational efficiencies \cite{bauer2021ripser, ctralie2018ripser, Čufar2020}. 

In this work, we take inspiration and combine the computational efficiencies of the Rips and Alpha filtrations to define a new filtration on Euclidean point cloud data. This construction, which we refer to as the Delaunay-Rips (DR) complex, may be viewed as a special case of the lazy-witness complex \cite{Silva2004TopologicalEU}. Instead of using the weak witness complex, we use the strong witness complex that is the Delaunay triangulation \cite{de2008weak}. We also establish that, generically, this filtration will inherit previously established stability results which guarantee that, roughly, certain small changes in the underlying data will result in small changes in the PD \cite{Chazal_gromov_stabililty, CohenSteiner2005StabilityOP}. That said, we further prove that, for certain configurations of points, the DR filtration suffers from an instability in which an arbitrarily small perturbation of the underlying data can result in a non-infinitesimal change in the PD. This instability is due to the instability in the Delaunay triangulation itself at degenerate configurations of points.

Much of the fundamental appeal of PH technologies is due to the numerous stability results that show the transformations sending data to diagrams and/or diagrams to vectorized topological representations are Lipschitz continuous transformations \cite{Skraba2020WassersteinSF}. However, it is often the case in practice that the topological representations derived from PDs---and used in data analysis and statistical model development---are not stable, despite the purported success of the models on which they’re based. Thus, we are further motivated to ask, to what extent does the instability observed in the DR filtration matter in practice? We interrogate this question empirically by performing systematic comparisons of the performance and costs of the Alpha, Rips, and DR filtrations on synthetic and real datasets and machine learning (ML) modelling tasks. We find the effect of the instability is negligible in these cases while the practical computational advantages enjoyed by DR can be significant. Thus, the contributions of this paper are:
\begin{itemize}
    \item We define the DR complex, which is a refinement of the Vietoris-Rips simplices on the data using the Delaunay triangulation as a backbone.
    \item We offer an algorithm and a Python implementation of the DR filtration.
    \item We document some empirical runtime comparisons for constructing the PD using DR with popular implementations of the Rips \cite{ctralie2018ripser} and Alpha \cite{cechmate} filtrations.
    \item We provide a straightforward proof of the stability of the Delaunay triangulation in some neighborhood of a generic point cloud, which establishes stability of the DR filtration on an appropriately chosen neighborhood of the data.
    \item We provide a rare proof of an instability in a filtration used for computing PH.
    \item We examine the practical impact of the instability on synthetic and real-world classification tasks and find the instability has limited impact on model performance. 
\end{itemize}

This paper is organized as follows: In Section \ref{sec:background} we briefly review the necessary mathematical preliminaries. Section \ref{sec:delrips} formally introduces the DR filtration, provides psuedo-code of the implementation we used to compute the DR filtration on point cloud data, and compares the empirical runtime of this algorithm to implementations of Rips and Alpha constructions across increasing point cloud size and dimension. In Section \ref{sec:stability} we discuss the stability properties of the DR filtration and demonstrate---through a by-hand calculation of a family of PDs---how a discontinuity in the transformation from point cloud to PD can arise given a perturbation in the underlying cloud. In Section \ref{sec:ml-comparisons} we report the result of several systematic comparisons of ML model performance trained on topological features derived from Alpha, Rips, and DR flirtations on synthetic and experimental data, including using random forest classifiers trained on persistence image (PI) vectorizations \cite{adams2017persistence} of persistence diagrams, and support vector machines trained on persistence statistics feature vectors derived from EKG time-series data \cite{Chung_frontiers_hr_ml}.

\section{Background}
\label{sec:background}
To extract topological features from a point cloud, one often begins by treating the points as the \textit{vertices} of a so-called \textit{simplicial complex} that consists of vertices, edges (pairs of vertices), triangles (sets of three vertices), tetrahedra (sets of four vertices), and higher-dimensional analogues (sets of $n > 4$ vertices), to construct an object with well-defined notions of ``shape.'' Our interest is in computing algebraic objects, namely persistent homology groups, that characterize topological (homological) invariants (e.g., connected components, holes, voids, etc.) across scales. We briefly review these notions here. For a deeper treatment of the mathematical foundations of these sections, we refer the reader to \cite{Edelsbrunner, hatcher2002algebraic}. For a more computational treatment with many practical considerations and examples, we recommend \cite{Roadmap}.

\subsection{Simplicial Homology}
\begin{defn}
    Let $K_0$ be a finite set and $\mathcal{P}(K_0)$ the powerset (i.e., the set of all subsets) of $K_0$. An \textbf{abstract simplicial complex built on $K_0$} is a collection, $K \subset \mathcal{P}(K_0)$, of non-empty subsets of $K_0$ with the properties that $\{v\} \in K$ for all $v \in K_0$, and if $\sigma \in K$ then $\tau \in K$ for all $\tau \subseteq \sigma$. 
\end{defn}
    In this paper we are concerned with simplicial complexes built from point clouds such that the points are identified with the singleton sets in the complex. In general, the elements of a simplicial complex, $K$, will be called \textbf{simplices}, while we may refer to sufficiently small subsets by other names. For example, we will refer to singleton sets as \textbf{vertices} of $K$, size-two sets as \textbf{edges}, etc. We say that a simplex has \textbf{dimension} $p$ or is a \textbf{$p$-simplex} if it has size $p+1$; so vertices are dimension 0 simplices, edges are dimension 1, triangles are 2-simplices, etc. We denote a $p$-simplex by $[v_0 v_1\ldots v_p]$ if it contains the 0-simplices $v_i, i=0,\ldots p$. $K_p$ denotes the collection of all $p$-simplices and the \textbf{$k$-skeleton} of $K$ is the union of the sets $K_p$ for all $p \in \{0,1,\dots,k\}$. If $\tau$ and $\sigma$ are simplices such that $\tau \subset \sigma$, then we call $\tau$ a \textbf{face} of $\sigma$. We say that $\tau$ is a face of $\sigma$ of \textbf{codimension} $q$ if the dimensions of $\tau$ and $\sigma$ differ by $q$. The \textbf{dimension} of $K$ is defined as the maximum of the dimensions of any of its simplices.
\begin{defn}
    Let $K$ be a simplicial complex and $p \geq 0$. A \textbf{$p$-chain} of $K$ is a formal sum of $p$-simplices in $K$ written as
    \[c = \sum a_i \sigma_i\]
    where $a_i \in \mathbb{F}$ is a field, and $\sigma_i$ are $p$-simplices. The $p$-chains then form a \textbf{vector space of $p$-chains} over a field $\mathbb{F}$, which we denoted $C_p(K)$.
\end{defn}

\noindent $C_p(K)$ and $C_{p-1}(K)$ are naturally related by a linear map called the \textbf{boundary map} that sends each $p$-chain to its boundary $(p-1)$-chain.
\begin{defn}
    Let $K$ be an $n$-dimensional simplicial complex and $\sigma = [u_0u_1\dots u_p]\in C_p(K)$. The $p$-\textbf{boundary map}, $\partial_p$ is defined on $p$-simplices to be
    \[\partial_p\sigma = \sum_{j=0}^{p}(-1)^j[u_0\dots\hat{u}_j\dots u_p]\]
    where the hat indicates that $u_j$ is omitted. Extending via linearity, $\partial_p:C_p(K) \to C_{p-1}(K)$ is further defined on any $p$-chain, $c = \sum_i {a_i}\sigma_i$, by
    \[\partial_p c = \sum_i a_i\partial_p\sigma_i.\]
\end{defn}
   \noindent Intuitively the chain of codimension-1 faces of a simplex encode the boundary of that simplex. Let's take a look at some examples. The boundary of the triangle $[abc]$ is the chain $[bc] - [ac] + [ab]$, consisting of the three edges that form its boundary. Extending to chains, the boundary of the two edges that meet at the vertex $b$ (namely $[ab]+[bc]$) would be $[b]- [a] + [c] - [b] = [c]-[a]$, reflecting the fact that the boundary of this 1-dimensional object consists only of its two terminal vertices. Connecting the chain groups via their boundary maps forms a \textbf{chain complex}
    \[\bm{0} \rightarrow C_n(K) \xrightarrow{\partial{n}} \dots C_{p+1}(K)\xrightarrow{\partial{p+1}} C_{p}(K)\xrightarrow{\partial{p}} C_{p-1}\xrightarrow{\partial{p-1}}\dots C_0(K) \rightarrow \bm{0},\]
    \noindent where the first map is the trivial linear map that sends the 0-vector in the trivial vector space, $\bm{0}$, to the 0-vector in $C_n(K)$. 
    
    Informally, a $p$-dimensional hole in a simplicial complex will be represented by a $(p-1)$-chain that could be (but isn't) the boundary of a $p$-chain. For example, as we have seen the boundryless 1-chain $[bc] + [ac] + [ab]$ is the boundary of a 2-chain (the 2-simplex $[abc]$). However, if $[abc]$ were not in the simplicial complex, we would be justified in saying the complex contained a hole enclosed by the edges of the missing triangle. To make these notions precise we introduce two subspaces of $C_p(K)$ that respectively encode all the $p$-chains that are without boundary, and those which are actually boundaries of $p+1$ chains. 
    
\begin{defn}
    A \textbf{$p$-cycle}, $\gamma \in C_p(K)$, is a $p$-chain with no boundary, i.e., $\partial_p \gamma = 0$. The collection of all $p$-cycles in a simplicial complex $K$ is denoted $Z_p(K)$ and forms a subspace of $C_p(k)$ because $Z_p(K)=\ker\partial_p.$
\end{defn}

\begin{defn}
    A \textbf{$p$-boundary}, $\beta \in C_p(K)$, is a $p$-chain that is the boundary of a $(p+1)$-chain, i.e., $\partial_{p+1} \sigma = \beta$ for some $\sigma \in C_{p+1}(K)$. The collection of $p$-boundaries in the simplicial complex $K$, denoted $B_p(K)$, is also a subspace of $C_p(K)$ since $B_p = \text{im}\ \partial_{p+1}.$
\end{defn}
With $p$-cycles and $p$-boundaries defined, we can formally define the object which captures representatives of holes in a simplicial complex: those cycles which are not boundaries. 
\begin{defn}
    The \textbf{$p$-th homology group} is
    \(H_p = Z_p(K)/B_p(K) = \ker\partial_p / \text{im}\ \partial_{p+1}.\)
    The \textbf{$p$-th Betti number} is the dimension of the quotient vector space, $\beta_p = \text{rank}\ H_p$.
\end{defn}
$H_p$ is then the vector space of equivalence classes of $p$-cycles, where two $p$-cycles are equivalent if they differ by a $p-$boundary. The fact that $H_p$ is well-defined relies on the fundamental result that ensures the boundary of a boundary chain must be empty, so that $B_p$ is actually a subspace of $Z_p$.
\begin{lem}
    (Fundamental lemma of homology \cite{Edelsbrunner}) The composition of any two consecutive boundary maps in the chain complex is trivial. That is,
    \[\partial_p \partial_{p+1} (\sigma) = 0\]
    for every integer $p$ and every $(p+1)$-chain $\sigma$.
\end{lem}

As suggested, the $p$-th Betti number corresponds to the number of $p$-dimensional ``holes" in the corresponding simplicial complex. In fact, the $0$-th Betti number counts the number of connected components, the $1$-st Betti number counts the number of (independent) loops, the $2$-nd Betti number counts the number of (independent) voids, etc.

\subsection{Persistent Homology}

Although simplicial homology is sufficient for capturing some intrinsic shape characteristics of a fixed simplicial complex, the natural question that arises with point cloud data is how to construct a simplicial complex from the points to provide a meaningful representation of latent structure in the cloud. In fact, there may be topological features of interest for one complex built on the data (for instance at some fine scale), and another set of topological features (at a larger scale) of interest from another complex. A promising approach to deal with this question is persistent homology, an extension of homology that considers a parameterized family of simplicial complexes, rather than just one. 

\begin{defn}
    Let $K_s$ be a simplicial complex for each $s \in \mathbb{R}$ such that $K_s \subset K_t$ for all $s \leq t$. We refer to such a collection of complexes as a \textbf{filtration}. When each $K_s \subseteq K$, for some fixed finite simplicial complex $K$, and $K_t = K$ for some $t \in \mathbb{R}$, we'll refer to the parameterized collection as a \textbf{filtration of} $K$.
\end{defn}
Note that a filtration on a finite simplicial complex $K$ necessarily only contains finitely many distinct subcomplexes, which we can relabel 
$$\emptyset := K_0 \subseteq K_{1} \subseteq K_{2} \subseteq \dots \subseteq K_{n} = K.$$
Here we have relabelled $K_i := K_{s_i}$ for the scales $s_{i} \leq s_{i+1}$, at which the subcomplexes change. Each simplex $\sigma \in K$ may also be assigned the minimum scale at which it appears in the filtration. In other words, one may define $f:K \to \R$ by $f(\sigma)=t$ if $\sigma \in K_t$ and $\sigma \notin K_s$ for any $s<t$. Thus, a filtration induces a so-called \textit{monotonic function} on the simplices of $K$, where $f(\sigma) \leq f(\tau)$ if $\sigma$ is a face of $\tau$. Conversely, any monotonic function $f: K \rightarrow \R$ with the property that $f(\sigma) \leq f(\tau)$ if $\sigma \subseteq \tau$, defines a filtration on $K$ by taking $K_s = f^{-1}(\infty, s].$ 

\begin{defn}
    Let $\emptyset \subseteq K_1 \dots \subseteq K_n = K$ be a filtration of $K$. The \textbf{$p$-th persistent homology groups} of the filtration are then formally defined as
    \[H_p^{i,j} = Z_p(K_i)/\left(B_p(K_j) \cap Z_p(K_i)\right).\]
    for $0 \leq i \leq j \leq n$. The corresponding \textbf{$p$-th persistent Betti numbers} are the ranks of these groups,
    \[\beta_p^{i,j} = \text{rank}\ H_p^{i,j}.\]
\end{defn}
An element of $H_p^{i,j}$ corresponds to a cycle in the filtration that persisted from $K_i$ to $K_j$ (i.e. a cycle in $K_i$ that did not become a boundary in $K_j$) and the content of all the $p$-th persistent homology groups of a filtration can be summarized in a dimension-$p$ persistence diagram which tracks the pairs of indices in the filtration at which homological features first appear and later disappear. It is common practice to say a feature is ``born" in complex $K_i$ and ``dies" in complex $K_j$ if the scale it became a cycle (without being a boundary) is in $K_i$ and earliest complex it becomes a boundary is $K_j$.

\begin{defn}
    Let $\mu_p^{i,j}$ be the number of $p$-dimensional classes born in complex $K_i$ that die entering complex $K_j$. Then
    \[\mu_p^{i,j} = (\beta_p^{i,j-1}-\beta_p^{i,j})-(\beta_p^{i-1,j-1}-\beta_p^{i-1,j}).\]
    The \textbf{$p$-persistence diagram} of the filtration given by $f: K \rightarrow \R$, denoted $\text{Dgm}_p(f)$, is a multiset of pairs $(s_i,s_j)$ in the extended real plane $\bar{\R}^2$ with multiplicity $\mu_p^{i,j}$. Each pair $(s_i,s_j)$ represents a nontrivial persistent homology class that is born in complex $K_i = K_{s_i}$ and dies upon entering complex $K_j = K_{s_j}$ because it merges with a homology class that was born before $K_i$.
\end{defn}

Now, we can define a metric we can use to find the distance between two $p$-persistence diagrams.
\begin{defn}
    Let $X$ and $Y$ be two $p$-persistence diagrams. Define $||x-y||_\infty := \max \{|x_1-y_1|,|x_2-y_2|\}$ for $x=(x_1,x_2) \in X, y=(y_1,y_2) \in Y$. We define the \textbf{bottleneck distance} between the diagrams as

    $$W_\infty(X,Y) = \inf_{\eta:X\to Y} \sup_{x\in X}||x-\eta(x)||_\infty$$
    where the infimum is taken over all bijections where $\eta$ can map to the diagonal if $X$ and $Y$ have different cardinalities. 
\end{defn}

A valuable property of persistence diagrams built from filtrations on a fixed complex $K$ is their stability with respect to changes in the monotonic function determining filtrations on $K$:
 \begin{thm}[\cite{Edelsbrunner}]\label{thm:edelstabilitythm}
        Let $K$ be a finite simplicial complex and $f,g: K \to \R$ two monotonic functions. For each dimension $p$, the bottleneck distance between the diagrams $\text{Dgm}_p(f)$ and $\text{Dgm}_p(g)$ satisfies
        $$W_\infty\left(\text{Dgm}_p(f),\text{Dgm}_p(g)\right) \leq ||f-g||_{\infty},$$
        where $||f-g||_{\infty} = \max_{\sigma \in K} |f(\sigma) - g(\sigma)|.$
    \end{thm}

We will also use a notion of distance between subsets of a metric space to control the distance between point clouds living in $\R^D$.
\begin{defn}\label{defn:hausdorff}
    Let $X$ and $Y$ be two non-empty subsets of a metric space $(M,d)$. We define their \textbf{Hausdorff distance} $d_H(X,Y)$ by
    \[d_H(X,Y):= \inf\{\varepsilon>0\ ;\ X \subseteq Y_\varepsilon\ \text{and}\ Y \subseteq X_\varepsilon\},\]
    where
    \[X_\varepsilon:= \bigcup_{x \in X} \{m \in M\ ;\ d(x,m) \leq \varepsilon\}.\]
\end{defn}
For finite subsets $X, Y \subset (M,d)$ (i.e., finite point clouds), the Hausdorff distance reduces to the maximum distance from a point in one set to the closest point in the other set. We prove a related statement in Lemma \ref{lem:hausdorffequaltopointdist}.

The formal notion of a filtration on a simplicial complex and its persistent homology groups defined in this section provide a means to extract multiscale structure from point cloud data, and thereby alleviate some of the concern about which complex may best capture structure in the cloud. Still, the methods of constructing a filtration from a point cloud are numerous and come with advantages and disadvantages that depend on the nature of the data.

\subsection{Vietoris-Rips and Alpha Complexes}
\label{rips}
One of the simplest and most commonly used methods to build a filtration on a finite point cloud, $X \subset (M,d)$, is to treat the points as vertices and add simplices at a scale determined by their diameter in the metric space.

\begin{defn}
    Let $X \subset \R^D$ be a point cloud and $\varepsilon \geq 0$. The \textbf{Vietoris-Rips complex} at scale $\varepsilon$ is defined as
    \[\text{VR}_{\varepsilon}(X) = \{\sigma \subseteq X\ |\ d(x,x') \leq 2\varepsilon,\ \forall x,x' \in \sigma\}.\]
\end{defn}

\noindent In other words, for a given scale $\varepsilon\geq 0$, if $d(x,x')\leq 2\varepsilon$ for $x,x' \in X$, one adds the $p$-simplex $\sigma = [x_0x_1\dots x_k]$ to the complex at the largest scale of any of its edges.

Although an algorithm to compute the Rips complex at any scale is simple to implement, constructing it on point cloud, $X$, with large numbers of points results in a computational challenge: eventually (at scales at and beyond half the diameter of the point cloud) the Rips complex will be equal to the powerset of $X$ and so will contain $2^{|X|}$ simplices. Moreover, the Rips complex will eventually contain simplices of all dimensions (up to the size of the point cloud minus 1), and so will contain homological information even beyond the dimension of the cloud (assuming the data lives in a finite-dimensional vector space). In practice this can be mitigated by imposing a restriction on the maximum dimension of simplices included in any complex in the filtration. Moreover, it is often the case in practice that, as the scale increases, additional simplices may appear in the complex that do not affect its homology \cite{lc-reduction, ADAMASZEK20171}. Filtration methods which avoid inclusion of ``extraneous'' simplices may be preferable for large point clouds \cite{sheehy2012linear, Guibas2007ReconstructionUW, Silva2004TopologicalEU}. Before defining examples of such methods, it will be helpful for the remainder of the paper to define when a Euclidean point cloud is in ``general position".

\begin{defn}
    A set of points in a $d$-dimensional Euclidean space is in \textbf{general position} if no $d+2$ of them lie on a common ($d-1$)-sphere.
\end{defn}

\noindent For example, this means for a set of points to be in general position in $\R^2$, no 4 of them can be co-circular. This condition ensures that each subset of $d+1$ points lie on a unique $d$-dimensional sphere which will ensure a unique Delaunay triangulation, defined below. 

Let $X \subset \mathbb{R}^D$ be a finite point cloud. Let $x \in X$ and define
$$V_x := \{p \in \R^D\ |\ d(p,x) \leq d(p, x')\ \forall x' \in X\}.$$
Each $V_x$ is called a \textbf{Voronoi cell} of $X$ and captures all points which are not closer to any other point in $X$ than $x$. Note that $\{V_x\}_{x \in X}$ forms a cover of $\R^D$. This cover is known as the Voronoi decomposition of $\R^D$ with respect to $X$. To construct the Delaunay triangulation from this cover, we define

$$\Del(X) := \{\sigma \subset X\ |\ \bigcap_{x \in \sigma} V_x \neq \emptyset\}.$$

It is known that $\Del(X)$ is itself a simplicial complex \cite{Edelsbrunner}. An $n$-simplex $\sigma \in \Del(X)$ will be referred to as a \textbf{Delaunay simplex}. We will use $\Del(X)$ as the underlying structure when defining the Delaunay-Rips complex in Section \ref{del-rips:def}. 

First we recall another commonly used filtration construction, known as the Alpha filtration, well studied for point clouds $X \subset \R^D$. We recall the definition given in III.4 of \cite{Edelsbrunner}. 

\begin{defn}
Let $\varepsilon\geq 0$ and let $S_x(\varepsilon) := V_x \cap B_x(\varepsilon)$, where $B_x(\varepsilon)$ is the $d$-dimensional ball of radius $\varepsilon$ centered on $x \in X$. The \textbf{Alpha complex} at scale $\varepsilon\geq 0$ is
\[\text{Alpha}_\varepsilon(X) = \{\sigma \subseteq X\ |\ \bigcap_{x \in \sigma} S_x(\varepsilon) \neq \emptyset\}.\]
\end{defn}

Note that since $S_x(\varepsilon) \subseteq V_x$, the set of 1-simplices of the $Alpha_{\epsilon}$ complex form a subcomplex of the 1-skeleton of the Delaunay triangulation.

By construction, $VR_{\varepsilon}(X)$ and $Alpha_{\varepsilon}(X)$ are simplicial complexes for all $\varepsilon \in \mathbb{R}$, and if $s \leq t$, both $Alpha_{s} \subseteq Alpha_{t}$ and $VR_{s} \subseteq VR_{t}$. Thus, each filtration construction yields an ordering on a set of simplices in a simplicial complex built on the point cloud $X$. For Vietoris-Rips, the scale of each simplex is determined by the distance between the farthest two vertices that define the simplex. However, Vietoris-Rips also assigns a non-zero weight to every subset of a set of vertices which is an exponentially slow computation in the number of vertices. Alpha, on the other hand, does not compute scales for every subset of the set of vertices. Rather, the scales assigned to simplices are determined by when the restricted epsilon balls on the Voronoi cells intersect. This additional computation is what we seek to avoid in our construction of the Delaunay-Rips complex in the subsequent section.

\section{The Delaunay-Rips Complex and Stability}
\label{sec:delrips}
\subsection{Definition and Construction} \label{del-rips:def}
Combining the Alpha and Rips constructions provides an alternative method of building a family of complexes on a point cloud $X \subset \mathbb{R}^d$. The idea is similar to the construction of the Delaunay-\u Cech complex defined in \cite{Bauer_2016} and can be seen as a special case of a lazy weak witness complex \cite{Silva2004TopologicalEU}. Delaunay-Rips utilizes the conceptual simplicity of the Vietoris-Rips complex while cutting down on the number of high dimensional and potentially extraneous simplices. The idea is that we build the Vietoris-Rips complex on $X$ but only add edges if the edges occur in the Delaunay 1-skeleton of the point cloud. The higher dimensional $p$-simplices are then added in the traditional Vietoris-Rips manner, i.e., if and only if their 1-skeletons appear.

\begin{defn}\label{defn:del-rips_complex}
The \textbf{Delaunay-Rips (DR) complex} for a given scale $\varepsilon\geq0$ is defined,
\[\text{DR}_{\varepsilon}(X) = \{\sigma \subseteq \Del(X)\ |\ d(x,x') \leq 2\varepsilon,\ \forall x,x' \in \sigma\}.\]
\end{defn}
Just as previously with Vietoris-Rips and Alpha filtrations, the Delaunay-Rips filtration can be thought of as a method for building a complex on a point cloud and assigning a scale, and thus a monotonic ordering, to the simplices in our complex built on $X$.

\begin{figure}[ht]
    \centering
    \includegraphics[width=\columnwidth]{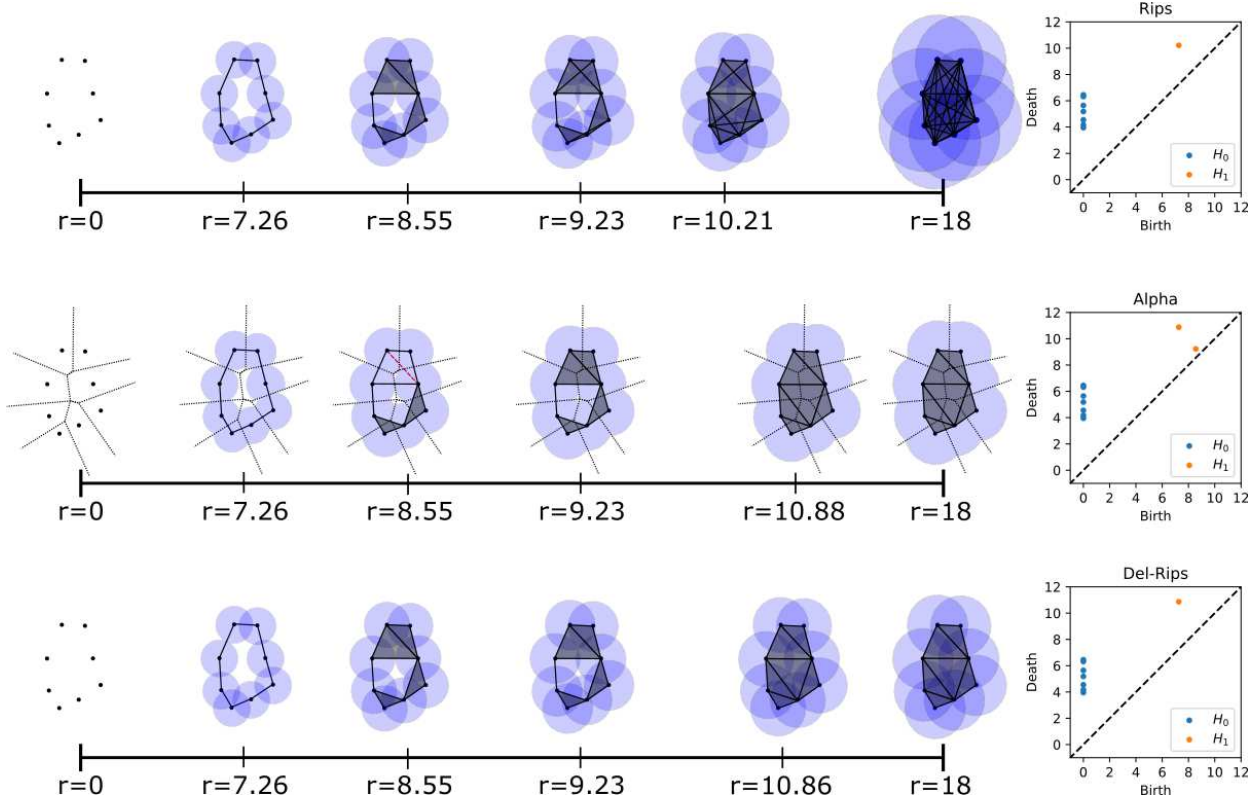}
    \caption{On the left, from top to bottom are the Rips, Alpha, and Delaunay-Rips filtrations on a point cloud with 8 points in $\R^2$ (black dots). On the right are the corresponding $H_0$ and $H_1$ persistence diagrams associated to each filtration. In each filtration the shaded circles have radii $r$ at each scale value, and edges and triangles are introduced according to definition of the given filtration at the specified scale. The dotted gray lines in the Alpha filtration show the boundaries of the Voronoi cells. In the PDs, the blue points represent the $H_0$ classes (connected components) and the orange points represent the $H_1$ classes (loops).}
    \label{fig:compare_filtrations_and_pds}
\end{figure}

Figure \ref{fig:compare_filtrations_and_pds} illustrates and compares examples of Rips, Alpha, and Delaunay-Rips filtrations, built on a cartoon point cloud in $\mathbb{R}^2$, and their associated $H_0$ and $H_1$ persistence diagrams. The Rips filtration produces a homology class that is born at scale value 7.26 and persists until dying at scale value 10.21. This is represented by the birth and death coordinates of the single orange point in the PD. Notice that although there are two loops in the topological surface of the overlapping circles at scale value 8.55, Rips does not capture the second loop in the upper portion because a triangle is added as soon as its 3 boundary edges appear. This is in contrast to the Alpha filtration which produces a homology class at scale value 7.26 that persists until scale value 10.88 and another $H_1$ class born at 8.55 and dies at scale value 9.23.  Recall that the Alpha complex at a particular scale value requires the balls centered at each 0-simplex of a $d$-simplex in $\R^D$ to be restricted to the Voronoi cells of the respective $0$-simplex. Hence, the $d$-simplex is only added to the complex when all of the restricted balls intersect. In our example, the dashed red edge--which appears at scale 8.55, along with the higher dimensional triangles it creates, in the Rips and DR filtrations---does not appear until scale value 9.23 in the Alpha filtration due to the aforementioned property of the Alpha complex. Particularly, that edge needs to wait until the radii of the balls equals the radius of the circle containing the 3 points that make up the simplex that edge lies on. 

The DR filtration produces an $H_1$ homology class that persists from scale value 7.26 to 10.86. In the DR construction, we add triangles as soon as their 3 edges appear in the complex: this is illustrated as a shaded triangle in the figure. Notice how at $r=18$ although many balls overlap, we have only added the simplices that appear in the Delaunay triangulation of the point cloud. The ending complex is combinatorially the exact same as the final Alpha complex. A notable comparison here is that the Del-Rips PD shows an $H_1$ class born at the same scale value as it was in the Rips filtration, but it persists for longer (due to the combinatorial cut down in simplices). Another notable comparison is that Del-Rips does not capture a second $H_1$ class in the PD like Alpha and the only $H_1$ class that is captured does not persist for quite as long as the corresponding $H_1$ class in Alpha.

As a remark, the reader may wonder if that second loop should be captured in the PD or not. Although there is a scale value at which the balls overlap in such a way as to reveal two loops in the data, it is not clear whether capturing the less persistent loop would be valuable in 
any particular application. Employing Alpha certainly captures that other loop well, but at the cost of computational efficiency as it seeks to balance the weights on the simplices across varying dimensions appropriately.  Section \ref{sec:ml-comparisons} partially addresses whether we can sacrifice the fidelity of Alpha to the underlying topological structure of the data for the advantage of a computational speed-up.

\subsection{Implementation and Runtime Analyses}
\label{sec:runtimes}

Here we show empirical results of the performance of computing persistence diagrams using the Delaunay-Rips filtration on varying datasets. We fix our field of coefficients to be $\Z_2$ when computing persistent homology groups. Algorithm \ref{alg:one}, which we have implemented in Python \cite{Amish_Mishra_Delaunay-Rips}, constructs the DR filtration across scales. This filtration then gets passed to the Persistent Homology Algorithms Toolbox (PHAT) \cite{BAUER201776} to construct the boundary matrix, reduce the boundary matrix, and extract the persistent pairs. Experiments were run on a computer with an Intel i7-10875H processor running at 2.3 GHz, 64 GB of RAM, and running Ubuntu 20.04.5 LTS. The runtimes are measured as the time between the data set being input into the corresponding algorithm and the persistence diagram being produced. For comparison, we choose the Python module Ripser \cite{ctralie2018ripser} and a Python Alpha implementation from the Python package Cechmate \cite{cechmate}. Figure \ref{fig:runtime_pts} shows how runtimes vary with the number of points being sampled from a noisy 2-sphere. Each data point on the plot is the median of 10 trials. The box-and-whisker plots on each data point show the max, min, and interquartile ranges of the runtimes. Similarly, Figure \ref{fig:runtime_dim} shows the runtimes as we vary the dimension of the sphere.

\SetKwComment{Comment}{/* }{ */}

\begin{algorithm}
\caption{An algorithm to compute Delaunay-Rips Filtration}\label{alg:one}
\KwIn{$P = (\mathbf{p_1},\dots,\mathbf{p_n})$, $dim$ = maximum homology dimension to compute}
\KwOut{Delaunay Rips Filtration}
triangulation $\leftarrow$ Delaunay($P$)\\
filtration $\leftarrow [([i], 0)]$ 
for $1\leq i\leq n$ \tcc{Add the 0-simplices into the filtration}
\For{each simplex $\in$ triangulation}{
    \For{$d \leftarrow 1$ \KwTo $dim+1$}{
        $faces \leftarrow$ all $d$-simplex subsets of current $simplex$\\
        \For{each face $\in$ faces}{
            \If{face $\notin$ filtration and $d == 1$}{
                $value \leftarrow$ distance($face[0]$, $face[1]$) \tcc{Calculate the euclidean distance}
                filtration.append(($face$, $value$)) \tcc{Append 1-simplices}
            }
            \If{face $\notin$ filtration and $dim > 1$}{
                find $subface$ of $face$ with greatest $value$\\
                filtration.append(($subface$, $value$)) \tcc{Add higher order simplices}
            }
        }
    }
}
\end{algorithm}

\begin{figure}[ht!]
    \centering
    \includegraphics[width=\columnwidth]{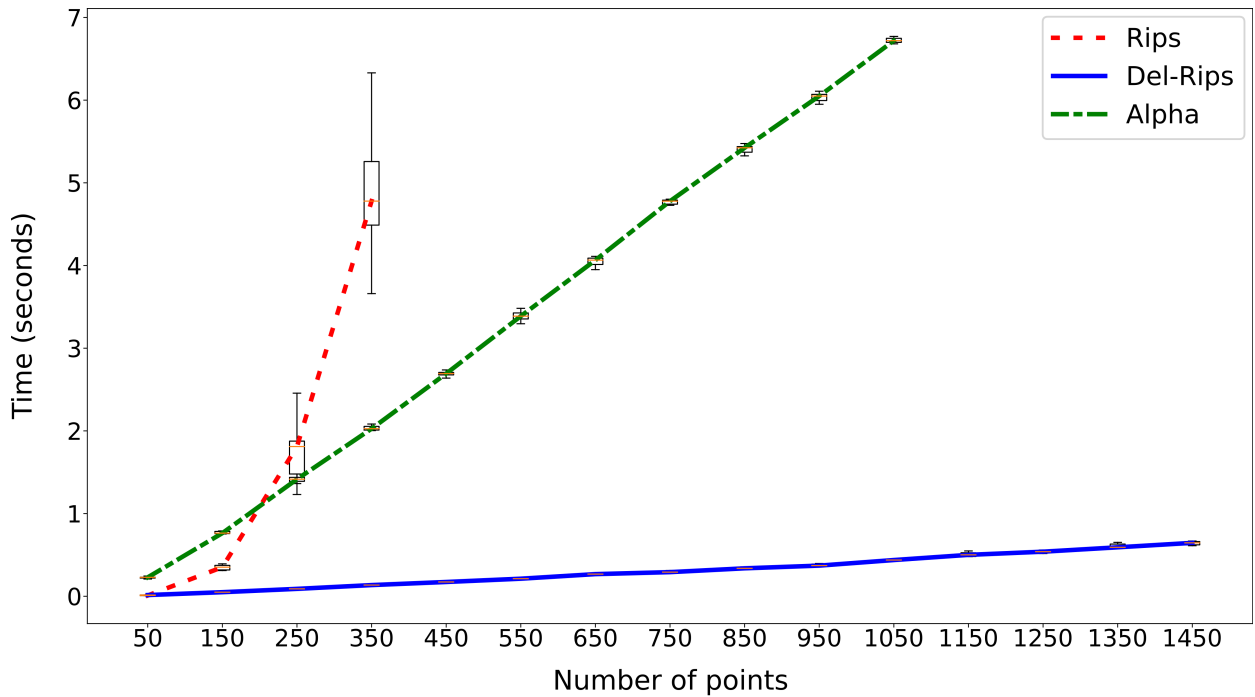}
    \caption{Runtime comparison of Rips, Alpha, and Delaunay-Rips as number of points are increased. Data set is taken from the surface of a 2-sphere of radius 1 with 0.1 noise. The maximum runtime allowed was 7 seconds.}
    \label{fig:runtime_pts}
\end{figure}

\begin{figure}[ht!]
    \centering
    \includegraphics[width=\columnwidth]{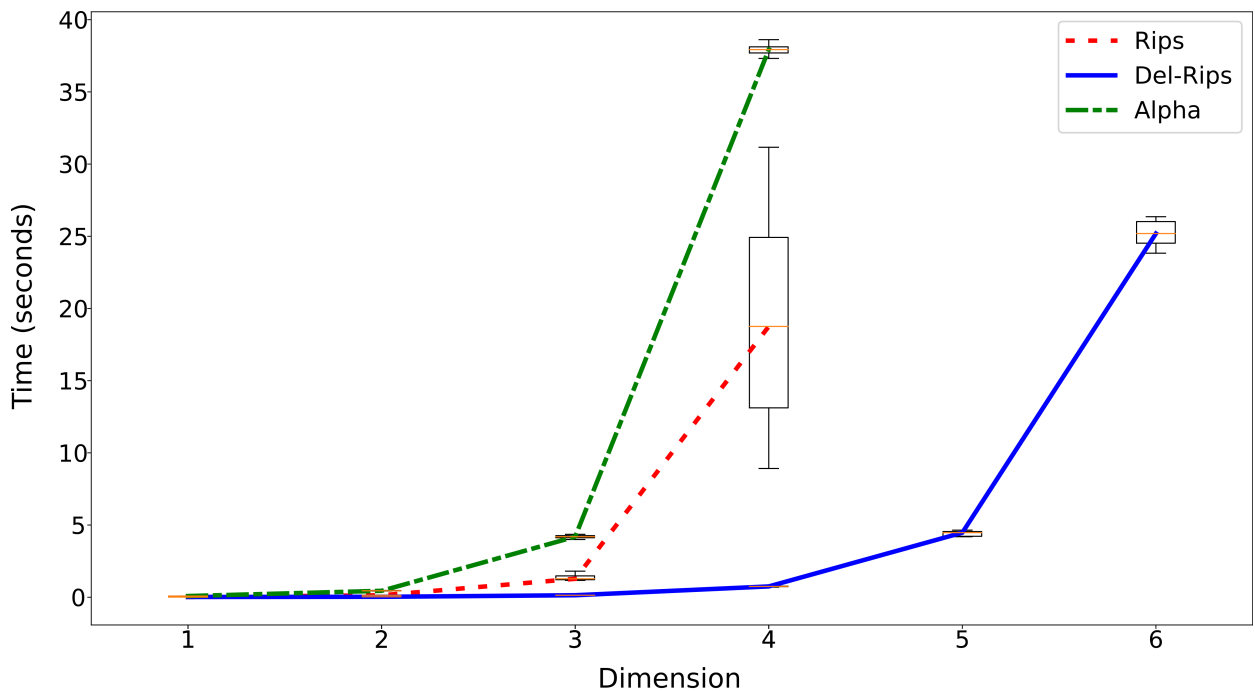}
    \caption{Runtime comparison of Rips, Alpha, and Delaunay-Rips as dimension of $d$-sphere increases. (Note that 1, 2, 3 on the x-axis correspond to 1-sphere, 2-sphere, 3-sphere). Data set is 100 points from the surface of a $d$-sphere of radius 1 with 0.1 noise. The maximum runtime allowed was 45 seconds.}
    \label{fig:runtime_dim}
\end{figure}

Notice that the Rips computation slows down in higher dimensions because we insist Ripser compute simplices up to the ambient dimension of the data set as is computed for Alpha and DR filtrations. For example, for data points on a 3-sphere (in $\R^4$), Rips, Alpha, and DR compute $H_0, H_1, H_2,$ and $H_3$ classes. Particularly, for Rips, this increase in dimension exponentially increases the number of simplices and thus increases the size of the boundary matrix which causes computational slowdown.

Notice that the Alpha computation slows down dramatically compared to DR even though both rely on computing the Delaunay triangulation. This is due to the computational overhead computing the scales at which multiple Voronoi cells intersect to assign simplex weights. Since DR avoids this computation, which is expensive in the Python implementation of Alpha in Cechmate, the practical speedup is significant. It is important to note that alternative implementations of Alpha (such as GUDHI \cite{gudhi:AlphaComplex}) may show speedups over the Python implementation of DR due to implementation details (e.g., the choice of implementation language), and while the the filtered complexes built by Alpha and Del-Rips are the same, we expect DR to enjoy some advantage in practice due to fast(er) calculation of simplex scales in optimized versions of these methods.

\subsection{Stability Properties of the Delaunay-Rips Complex}
\label{sec:stability}

Towards understanding the impact a perturbation of the underlying point cloud data has on the resulting DR persistence diagram, we adopt the notion of an $\varepsilon$-perturbation of a point cloud from \cite{Weller1997StabilityOV}. If $P = \{p_1,\dots,p_n\} \subset \R^D$ denotes a point cloud we let $P' = \{p_1',\dots,p_n'\}$ denote a perturbation of $P$, where we imagine placing $\varepsilon < \left(\min_{p_i,p_j \in P}\{d(p_i,p_j)\}\right)/2$ balls around each point $p_i \in P$ and selecting from each ball a perturbed point $p_i' \in P'$. 
    \begin{defn}\label{defn:epsilon_perturb}
        We call $P' = \{p_1', \ldots, p_n'\}$ an $\varepsilon$-perturbation of $P = \{p_1, \ldots, p_n\}$ if for each $p \in P$ there is exactly one $p' \in P'$ such that $d(p,p') < \varepsilon$ and, conversely, $p$ is the only point in $P$ within $\varepsilon$ of $p'$. We refer to each $p$, $p'$ as a perturbation pair.
    \end{defn}
     In other words, there is a bijection $\rho: P \rightarrow P'$ associating to each point $p_i \in P$ a point $p_i' \in P'$, so that $d(p_i, \rho(p_i)) < \varepsilon$.
    
\begin{lem}\label{lem:hausdorffequaltopointdist}
        Given that $P'$ is an $\varepsilon$-perturbation of a $P \subset \R^D$, there exists a perturbation pair $p_i \in P$ and $p_i' \in P'$ such that
        \[d_H(P,P') = d(p_i,p_i').\]
    \end{lem}
    \begin{proof}
                
        First, identify the perturbation pair, $x \in P$ and $x' \in P'$ which are farthest from one another among all pertubation pairs and let
        \begin{equation}\label{eqn:delta_defined}
             \delta := d(x,x') = \max_{p \in P, p'\in P'}\left(d(p,p')\right).
         \end{equation}
         By such a choice, we ensure that for any $p'\in P'$, it must be that $p' \in P_\delta$ since $d(p,p') \leq d(x,x') = \delta.$ 
         Similarly, for any $p \in P$, $d(p,p') \leq d(x,x') = \delta$ and so $P \subseteq P'_\delta$. Therefore $d_H(P,P') \leq \delta$ by definition \ref{defn:hausdorff}. 
         
         Now, assume $\lambda < \delta$. For any $p' \in P'$
         \[d(x,p') \geq \min_{p' \in P'}\left(d(x,p')\right) = d(x,x') = \delta > \lambda.\]
         Therefore $x \notin P'_\lambda$ since it is not in any $\lambda$-ball around the points of $P'$, therefore $d_H(P,P') > \lambda$ for every $\lambda < \delta$. Thus, $d_H(P,P') = \delta = d(x,x')$. 
        That is to say, the Hausdorff distance between $P$ and $P'$ is exactly equal to the largest distance a point in $P$ was perturbed within an $\varepsilon$-perturbation.
    \end{proof}

    To leverage theorem \ref{thm:edelstabilitythm} we will pursue conditions that ensure the underlying Delaunay triangulation doesn't change under a perturbation of the points. We say that a finite $P \subset \mathbb{R}^D$ and an $\varepsilon$-perturbation, $P'$ have the same Delaunay triangulation with respect to the perturbation pairing if $\sigma = [p_{i_0}\ldots p_{i_k}] \in \text{Del}(P)$ if and only if $\sigma' = [p_{i_0}'\ldots p_{i_k}'] \in \text{Del}(P')$ for all $0 \leq k \leq D$, where $p_{i_l}'$ is the unique point in $P'$ uniquely associated to $p_{i_l}$, with $d(p_{i_l},p_{i_l}')< \varepsilon$. As abstract simplicial complexes, $\text{Del}(P)$ and $\text{Del}(P')$ are indistinguishable, since the association of perturbation pairs induces a bijection between simplices. Thus we will not distinguish the associated simplices in these complexes.
    
    \begin{thm}\label{thm:hausdorffstability}
        Let $P \subset \R^D$ be a point cloud and $P'$ an $\varepsilon$-perturbation of $P$ with the same Delaunay triangulation with respect to the perturbation pairing; call it $K$. Let $f_{P}, f_{P'}: K \to \R$ be monotonic functions defined by assigning the Delaunay-Rips scales to the simplices of $K$ as viewed in $P$ and $P'$, respectively. For each dimension $p$, the bottleneck distance between the persistence diagrams is bounded from above by twice the Hausdorff distance between the point clouds $P$ and $P'$:
        \[W_\infty\left(\text{Dgm}_p(f_{P}), \text{Dgm}_p(f_{P'})\right) \leq 2d_H(P,P').\]
    \end{thm}
    \begin{proof}
        We will leverage the result of Theorem \ref{thm:edelstabilitythm}. 
        Since $P$ is a finite point cloud, there will be a simplex $\sigma \in K$ such that
        \[||f_{P}-f_{P'}||_{\infty} = |f_{P}(\sigma) - f_{P'}(\sigma)|.\]
        Since the Delaunay-Rips scale of a simplex is determined by the two points that are the farthest apart in the simplex, there exists $p, q \in P$ and $r', s' \in P'$, all of which are 0-simplices in $\sigma$, such that $f_P(\sigma) = d(p, q)$ and $f_{P'}(\sigma) = d(r',s')$. It could be that $p, q$ and $r', s'$ form two perturbation pairs (e.g., $p'=r'$ and $q'=s'$) but this need not hold in every case. For instance, the perturbation of the points $p$ and $q$ that determine the scale of $\sigma$ in $\Del(P)$ may have moved to $p'$ and $q'$ (both of which are necessarily in $\sigma$ in $\Del(P')$) closer together than the points $r'$ and $s'$ that determine the scale of $\sigma$ in $\Del(P')$  We consider these two cases separately.
        
        Case 1: $p'=r'$ and $q' = s'$. Without loss of generality, assume 
        \begin{equation*}
        d(p,q) > d(p',q').
        \end{equation*} 

        Then 
        \begin{align*}
        |f_{P}(\sigma) - f_{P'}(\sigma)| &= |d(p,q) - d(r',s')| \\
        &= |d(p,q) - d(p',q')|\\
        &= d(p,q) - d(p', q') \\
        &\leq d(p,p') + d(p',q') + d(q',q) - d(q',p')\\
        &= d(p,p') + d(q',q)\\
        &\leq d(x,x') + d(x,x')\\
        &= 2 d_H(P,P'),
        \end{align*}
        where $x \in P$ and $x' \in P'$ are the paired points farthest from one another in $P$ and $P'$ as in Lemma \ref{lem:hausdorffequaltopointdist}. 
        Case 2: $p' \neq r'$ or $q' \neq s'$. As mentioned,
        \begin{equation}\label{eqn:34biggerthan12}
            d(r', s') \geq d(p',q')
        \end{equation}
        \begin{equation}\label{eqn:12biggerthan34}
            d(p, q) \geq d(r,s),
        \end{equation}
        by assumption of which pair of vertices in $\sigma$ determine the scales of in $\Del(P)$ and $\Del(P')$.
        Now, there are two sub-cases to consider.
        
        Case 2a: \(d(p,q) \geq d(r',s')\). Using Equation \ref{eqn:34biggerthan12} we have
        \[|f_{P}(\sigma) - f_{P'}(\sigma)| = |d(p,q) - d(r',s')| = d(p,q) - d(r', s') \leq d(p,q) - d(p', q')\]
        and the rest follows by the same argument as in Case 1.
        
        Case 2b: \(d(p,q) < d(r',s')\). Using Equation \ref{eqn:12biggerthan34} we have
        \[|f_{P}(\sigma) - f_{P'}(\sigma)| = |d(p,q) - d(r',s')| = d(r', s') - d(p,q) \leq d(r',s') - d(r, s)\]
        and the rest follows as in Case 1 with appropriate relabeling.
    
        Finally, by applying Theorem \ref{thm:edelstabilitythm}, we conclude that the bottleneck distance between the persistence diagrams associated with $P$ and $P'$, respectively, is bounded from above by twice the Hausdorff distance between the two point clouds:
        \[W_\infty\left(\text{Dgm}_p(f_{P}), \text{Dgm}_p(f_{P'})\right) \leq ||f_{P} - f_{P'}||_\infty \leq 2d_H(P,P').\]   
    \end{proof}

We see that when the underlying Delaunay triangulation on our point clouds is fixed (with respect to the perturbation pairing), assigning scales to simplices using the DR algorithm guarantees stability of the corresponding persistence diagram. Although it is known that for a point cloud, $P$ in general position, there exists a sufficiently small $\varepsilon > 0$ such that every  $\varepsilon$-perturbation $P'$ will have the same Delaunay triangulation  \cite{boissonnat2013stability}, the size of $\varepsilon$ may be quite small, which casts doubt on the practical utility of the above stability result for real applications in which there is measurement uncertainties.  In the next section, we explore what can happen when the underlying Delaunay triangulation does change as a result of perturbing data points. 

\subsection{Persistence Diagram Instability} \label{instability_section}
The DR construction gains computational efficiency at the cost of stability. We demonstrate a simple, yet clear example of how a discontinuity in the transformation from data to diagram can arise under a perturbation of the underlying data. In Figure \ref{fig:4pt_PDs_transformation}, moving from left to right, we imagine moving the right-most point (in red) to the right towards the unique inscribing circle for the other three points. When the right-most point is inside the circle, there is an $H_1$ class with non-zero persistence. This class disappears immediately when the 4 points lie on the same circle. Informally, this means that DR sees the 4 point form a loop in the first two stages and then immediately loses sight of the loop when the points become cocircular. In the figure, we have marked the Delaunay triangulation of the points to showcase the position of the right-most (red) point at which an edge flip occurs (namely when all four points lie on the same circle). What we are seeing is point clouds that are very similar in structure visually, but are producing very different persistence diagrams. An arbitrarily small perturbation of the right-most point to inside the circle gives a very different persistence diagram from an arbitrarily small perturbation to outside the circle. We now proceed to formally prove the instability of the persistence diagrams associated with this particular configuration of points.

\begin{figure}[ht]
    \centering
    \includegraphics[width=\columnwidth]{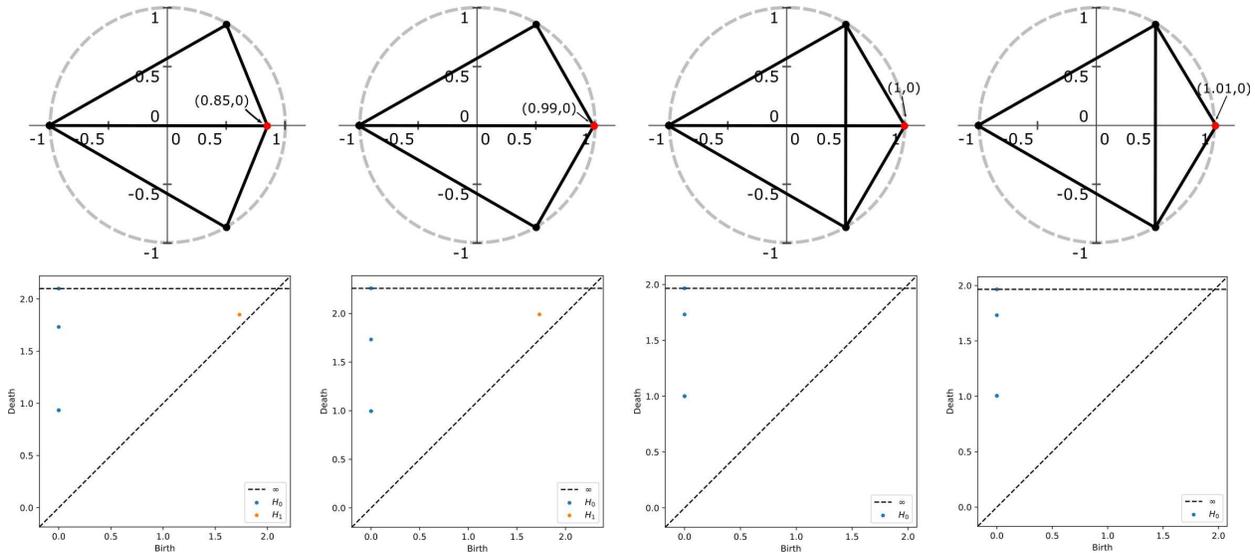}
    \caption{Persistence diagrams of 4 point example as the right-most point moves horizontally to the right on the $x$-axis. Notice that in the first two stages (from the left), there is an $H_1$ class with non-zero persistence in the diagrams. The $H_1$ class disappears in the last two stages.}
    \label{fig:4pt_PDs_transformation}
\end{figure}

\begin{lem}\label{lem:boundary_mat_for_instability}
    Let $P' = \{(-1,0),(\frac{1}{2},\frac{\sqrt{3}}{2}),(\frac{1}{2},-\frac{\sqrt{3}}{2}),(1-x,0)\}$ with $0< x < \delta < 2-\sqrt{3}$. Using the Delaunay-Rips complex to construct a filtration on this point cloud, there is only one $H_1$ homology class with non-zero persistence.
\end{lem}
\begin{proof}
We have $P' = \{(-1,0),(\frac{1}{2},\frac{\sqrt{3}}{2}),(\frac{1}{2},-\frac{\sqrt{3}}{2}),(1-x,0)\}$ with $0< x < \delta < 2-\sqrt{3}$.
Our filtration has 4 key scale values, $t = 0<\sqrt{1-x+x^2}< \sqrt{3}< 2-\delta$ as shown in Figure \ref{fig:4pt_filtration}. 

\begin{figure}[ht]
    \centering
    \includegraphics[width=\columnwidth]{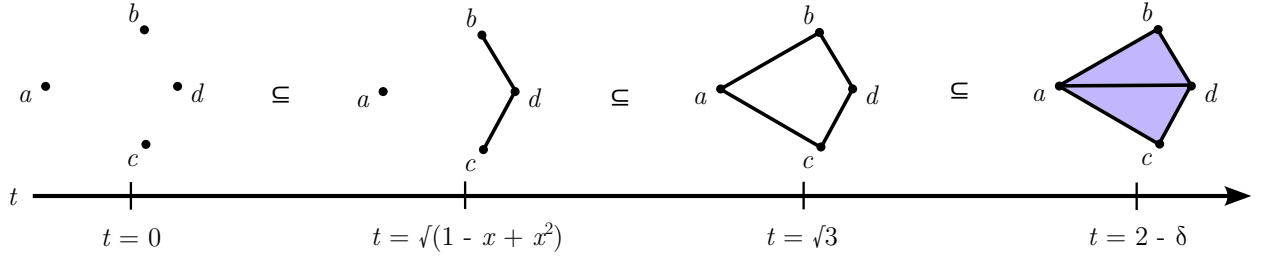}
    \caption{Filtration with 4 key scale values where $0<x<\delta<2-\sqrt{3}.$ }
    \label{fig:4pt_filtration}
\end{figure}

We construct a boundary matrix $B$ with entries from the field $\Z_2$ and reduce it to $\overline{B}$ using the standard reduction algorithm found in Chapter VII of \cite{Edelsbrunner}, which computes the pairing of simplices which respectively give rise to and kill off persistent homology classes in the $H_0$ and $H_1$ persistence diagrams. The ordering of the columns of $B$ is determined by the ordering of the simplices given by the DR filtration, while ensuring each simplex appears after its faces. 

\begin{equation*}
\begin{split}
B=\bordermatrix{ 
    &   a & b & c & d & bd & cd & ab & ac & ad & abd & acd\cr
    a & 0 & 0 & 0 & 0 &  0 &  0 &  1 &  1 &  1 &   0 &   0 \cr
    b & 0 & 0 & 0 & 0 &  1 &  0 &  1 &  0 &  0 &   0 &   0\cr
    c & 0 & 0 & 0 & 0 &  0 &  1 &  0 &  1 &  0 &   0 &   0\cr
    d & 0 & 0 & 0 & 0 &  1 &  1 &  0 &  0 &  1 &   0 &   0\cr
   bd & 0 & 0 & 0 & 0 &  0 &  0 &  0 &  0 &  0 &   1 &   0\cr
   cd & 0 & 0 & 0 & 0 &  0 &  0 &  0 &  0 &  0 &   0 &   1\cr
   ab & 0 & 0 & 0 & 0 &  0 &  0 &  0 &  0 &  0 &   1 &   0\cr
   ac & 0 & 0 & 0 & 0 &  0 &  0 &  0 &  0 &  0 &   0 &   1\cr
   ad & 0 & 0 & 0 & 0 &  0 &  0 &  0 &  0 &  0 &   1 &   1\cr
  abd & 0 & 0 & 0 & 0 &  0 &  0 &  0 &  0 &  0 &   0 &   0\cr
  acd & 0 & 0 & 0 & 0 &  0 &  0 &  0 &  0 &  0 &   0 &   0\cr
    }\\
\overline{B}=\bordermatrix{ 
    &   a & b & c & d & bd & cd & ab & ac & ad & abd & acd\cr
    a & 0 & 0 & 0 & 0 &  0 &  0 &  1 &  0 &  0 &   0 &   0 \cr
    b & 0 & 0 & 0 & 0 &  1 &  1 &  1 &  0 &  0 &   0 &   0\cr
    c & 0 & 0 & 0 & 0 &  0 &  1 &  0 &  0 &  0 &   0 &   0\cr
    d & 0 & 0 & 0 & 0 &  1 &  0 &  0 &  0 &  0 &   0 &   0\cr
   bd & 0 & 0 & 0 & 0 &  0 &  0 &  0 &  0 &  0 &   1 &   1\cr
   cd & 0 & 0 & 0 & 0 &  0 &  0 &  0 &  0 &  0 &   0 &   1\cr
   ab & 0 & 0 & 0 & 0 &  0 &  0 &  0 &  0 &  0 &   1 &   1\cr
   ac & 0 & 0 & 0 & 0 &  0 &  0 &  0 &  0 &  0 &   0 &   1\cr
   ad & 0 & 0 & 0 & 0 &  0 &  0 &  0 &  0 &  0 &   1 &   0\cr
  abd & 0 & 0 & 0 & 0 &  0 &  0 &  0 &  0 &  0 &   0 &   0\cr
  acd & 0 & 0 & 0 & 0 &  0 &  0 &  0 &  0 &  0 &   0 &   0\cr
    }.
\end{split}
\end{equation*}
By computing the scales of each simplex, the persistence pairs for the $H_0$ class with their persistence diagram coordinate (birth/death pair) are found to be
$$(a,N/A): (0,\infty)$$
$$(b,ab): (0,\sqrt{3})$$
$$(c, cd): (0,\sqrt{1-x+x^2})$$
$$(d, bd): (0,\sqrt{1-x+x^2}).$$
Likewise the $H_1$ persistence pairs are
$$(ad, abd): (2-x, 2-x)$$
and
$$(ac, acd): (\sqrt{3}, 2-x).$$
The only $H_1$ class with non-zero persistence is $(\sqrt{3}, 2-x).$
\end{proof}

\begin{thm}
    Let ($\mathcal{P}$, $d_{H}$) be the space of point clouds equipped with the Hausdorff metric and let ($\mathcal{D}$, $W_\infty$) be the space of persistence diagrams equipped with the bottleneck metric. Let
    $$\text{Pers}_1: \mathcal{P} \to \mathcal{D}$$
    where $\text{Pers}_1(P)$ is the persistence diagram of the $H_1$ classes of the point cloud $P$ constructed using the Delaunay-Rips complex. This map is discontinuous.
\end{thm}
\begin{proof}
    Let $P \in \mathcal{P}$ be $P = \{(-1,0),(\frac{1}{2},\frac{\sqrt{3}}{2}),(\frac{1}{2},-\frac{\sqrt{3}}{2}),(1,0)\}.$ Note that the points all lie on the unit circle, so the Delaunay 1-skeleton has an edge between every pair of points (See the third frame in Figure \ref{fig:4pt_PDs_transformation}). For this configuration of points, the vertical edge between the points $(\frac{1}{2}, \frac{\sqrt{3}}{2})$ and $(\frac{1}{2},-\frac{\sqrt{3}}{2})$ appears at the exact same scale value as the cycle formed by all four points. Therefore in the DR filtration, the $H_1$ class whose boundary is the four outer edges dies at the same time as it is born.

    Fix $\varepsilon=0.1$. We now show that for any $\delta > 0$, there exists $P' \in \mathcal{P}$ such that $d_{H}(P,P')< \delta$, but $W_\infty\left(\text{Pers}_1(P), \text{Pers}_1(P')\right) \geq \varepsilon.$ Take $P' = \{(-1,0),(\frac{1}{2},\frac{\sqrt{3}}{2}),(\frac{1}{2},-\frac{\sqrt{3}}{2}),(1-x,0)\}$ with $0<x < \min\{\delta, \frac{2-\sqrt{3}}{2}\}$. This is a small perturbation of $P$ gotten by pushing the point $(1,0)$ inside the unit circle, thereby putting the points in general position (See the second frame in Figure \ref{fig:4pt_PDs_transformation} for an example). It is straightforward to compute the Hausdorff distance $d_H$ between $P$ and $P'$ as
    $$d_H(P,P')=x<\delta.$$
    Recall that $\text{Pers}_1(P)$ has no $H_1$ class with non-zero persistence. Thus, to compute $W_\infty\left(\text{Pers}_1(P),\text{Pers}_1(P')\right)$, we must match the $H_1$ class of $\text{Pers}_1(P')$ with the diagonal. The $H_1$ class of $\text{Pers}_1(P')$ has birth $\sqrt{3}$ and death $2-x$ as calculated in Lemma \ref{lem:boundary_mat_for_instability}. Using the max norm, we find
    $$W_\infty(\text{Pers}_1(P),\text{Pers}_1(P')) = (2-x)-(\sqrt{3}) \geq 2-\frac{2-\sqrt{3}}{2} -\sqrt{3} \geq 0.1 = \varepsilon.$$
    So $\text{Pers}_1$ is discontinuous at $P$.
\end{proof}

As a remark, note that any metric on the space of point clouds that is bounded above by the Hausdorff distance will produce this discontinuity (for example, the Gromov-Hausdorff distance). This gives us insight into when the DR construction of the persistence diagram may experience an instability---namely when points are not in general position.         

Thus far, we have mathematically shown theoretical stability of the persistence diagram of a data set whose Delaunay triangulation does not change under the influence of a perturbation of the point locations. However, it should not be expected that in real-world applications the conditions guaranteeing stability will be met, and we have shown in this section that when the underlying Delaunay triangulation does change, the degree of change between diagrams may not be controlled by the degree of change in the underlying data. We are thus led to ask, to what extent does such an instability matter in practice?

\section{Machine Learning Model Performance using Rips, Alpha, and Delaunay-Rips Filtrations}
\label{sec:ml-comparisons}
    Although we have special cases where instability in the PD may arise (as shown in Section \ref{instability_section}), we show here that this instability may have little impact in applications. We demonstrate the robustness of DR for machine learning in a synthetic data context and a real data context. For the synthetic data, we work with random forest classifiers for a multi-class classification task. For the real data, we train support vector machines with linear kernel for a binary classification task. 

\subsection{Classification of Synthetic Shape Data}
\label{sec:ml-comparisons-synthetic}
    To test the robustness of DR to changes in the degree of perturbation to the locations of points in point cloud data, we develop ML classification models using Rips, Alpha, and DR filtrations on point clouds generated by randomly sampling various manifold and adding random noise to perturb the points. Although the persistence diagrams produced using the DR filtration enjoy stability when the underlying Delaunay triangulation is unchanged (see Section \ref{sec:stability}), in reality, the underlying Delaunay triangulation of the point is expected to change even for modest levels of noise.

    Figure \ref{fig:6shapeclasses} shows the general pipeline for our experiment. We generated 100 point clouds consisting of 500 data points for each of 6 shape classes (circle, sphere, torus, random, three clusters, and three clusters within three clusters) as subsets of $\R^3$. For a fixed level of noise $\nu$, points in each cloud were perturbed by randomly chosen vectors from the ambient space with maximum magnitude equal to $\nu$. For each cloud, we computed the 0, 1, and 2-persistence diagrams using Rips, Alpha, and DR filtrations and then vectorized the resulting diagrams using PIs from the Python package ``persim'' in the scikit-tda library \cite{persim}.  The resulting feature vectors were then used to train a random forest classifier using the implementation in scikit-learn \cite{scikit-learn}. To evaluate and compare the different filtrations, we computed the median classification accuracy of the trained models on held out data using 10-fold cross validation. 
    
    \begin{figure}[ht]
    \centering
    \includegraphics[width=\columnwidth]{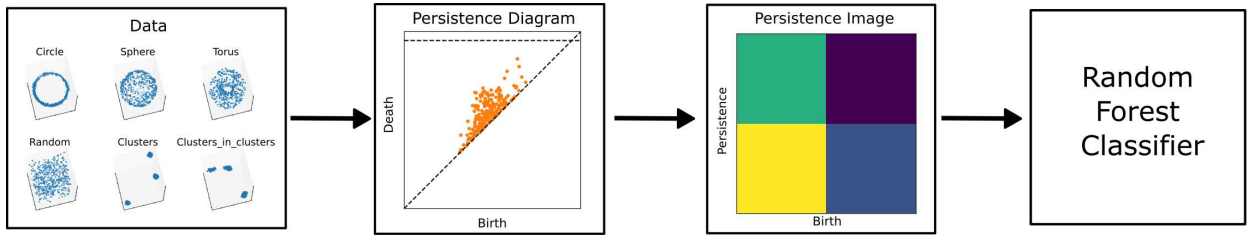}
    \caption{The first step is to generate data points based on 6 shape classes. The next step is to compute the persistence diagram for each data set for each dimension up to 3 (the image above is the $H_1$ persistence diagram for 500 random points in 3 dimensions.) Then, the diagrams are turned into PIs with a $2\times 2$ resolution grid. Finally, we flatten the image into a vector and train a random forest classifier to distinguish the 6 shape classes.}
    \label{fig:6shapeclasses}
    \end{figure}
    
    Fixing a homology class and fixing a noise level $\nu$ for the perturbation of the point cloud, we determine the birth and persistence ranges of persistence pairs produced over all filtration methods and over all samples, and we segment this region into $2\times2$ resolution persistence images. For example, we iterate through all $H_1$ persistence diagrams for all three of our filtration methods that were produced using $\nu=0.20$. Then, we find the maximum birth range and persistence range and use those values to set a $2\times2$ resolution grid to produce the PIs. The purpose of doing this was to ensure that if we compared pixels of the PIs corresponding to different filtration methods, we would make a fair comparison because corresponding persistence pairs would land in corresponding pixels (we leverage this design in Figure \ref{fig:heatmaps} when comparing feature importance).
    
    Figure \ref{fig:accuracy_v_noise} shows the median accuracy of our model for increasing noise levels (we plotted box-whisker plots to show the spread of the accuracy from 10-fold cross validation). As a baseline comparison, note that if our ML model was randomly classifying the test data, we expect to see accuracy of $1/6 \approx 17\%$. Since we are seeing over $70\%$ median accuracy for each noise level, our model truly is finding distinguishing features between the 6 shape classes.
    
    Notably, as is evident in Figure \ref{fig:accuracy_v_noise}, the degradation of ML model performance with increasing noise levels is very similar between DR, Alpha, and Rips filtrations. We do note that the median classification accuracy using Alpha filtrations to generate persistence diagrams is slightly better than the other methods across all noise levels. However, except for a single noise level (.15), all median accuracies are within the ranges of the other two methods. 
    
    \begin{figure}[ht]
        \centering
        \includegraphics[width=\columnwidth]{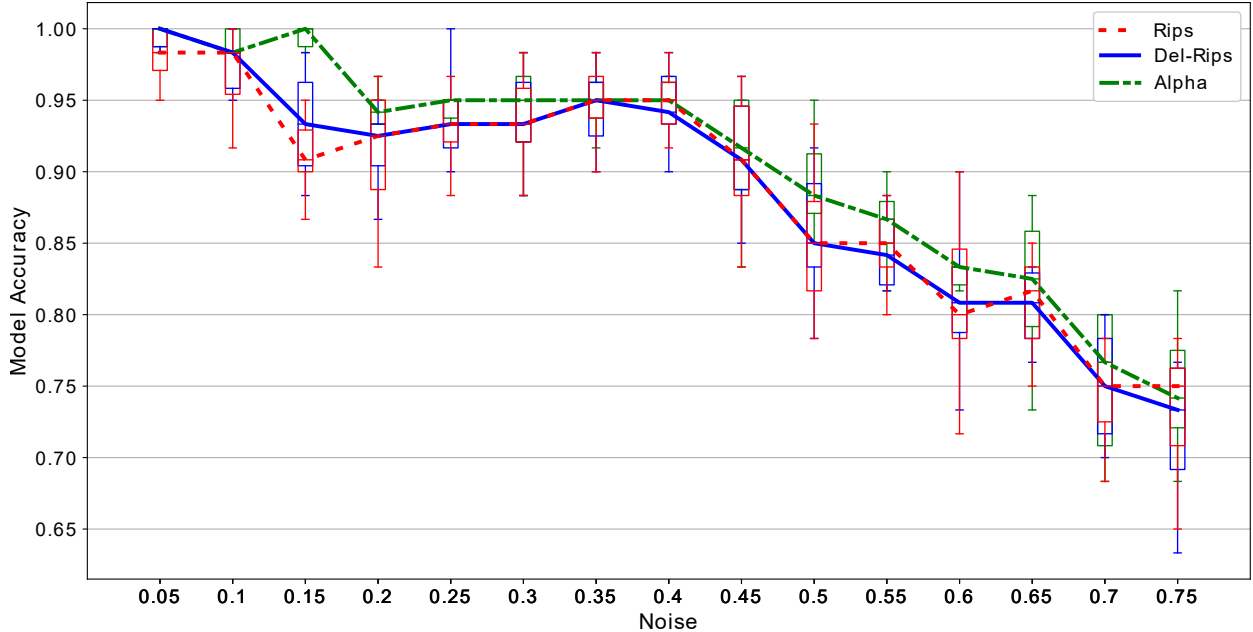}
        \caption{A plot of the median random forest classification accuracy computed on Rips, DR, and Alpha based persistence diagrams across varying noise levels for the original point clouds. The box-and-whisker plots for each noise level show the range of the accuracy across 10-fold cross validation.}
        \label{fig:accuracy_v_noise}
    \end{figure}
    
    In addition to model accuracy, an important consideration is which features are being used to achieve said accuracy. Identifying salient features provides a level of model interpretability \cite{ribeiro2016should}, and may guide the modeller to methods to reduce the dimension of the input, which may further improve model performance. To further compare the three filtration methods, we compared the most important topological features learned during training as determined by a random forest classifier. We began with the PDs produced from the $\nu=0.20$ noisy data. We generated PIs with the following resolutions for each persistence diagram:
    \begin{itemize}
        \item $H_0$ diagram resolution: $5 \times 1$
        \item $H_1$ diagram resolution: $5 \times 5$
        \item $H_2$ diagram resolution: $5 \times 5$
    \end{itemize}
    
   As a result, we produced 55-dimensional feature vectors for each of the 600 samples (100 samples each of 6 classses). We trained a random forest classifier with the same parameters as before, once using a train-test split of 70-30, and used the built-in assessment \cite{scikit-learn} of the Gini importance to quantify feature importance. The Gini importance of a feature is calculated as the amount of reduction to the Gini index \cite{Krzywinski2017PointsOS} brought by that feature. Thus, the higher the Gini importance, the more important the feature is for making classification decisions. Figure \ref{fig:heatmaps} shows heatmaps indicating feature importance. Notice how for the different filtrations we used to obtain the PDs, the PIs generated have similar corresponding pixels that the ML model found important. For example, among the 5+25+25 pixel feature space for the DR-based random forest classifier, the most important features were the bottom-left (dark blue) $H_2$ pixel with Gini importance 0.19 and the bottom $H_0$ pixel with Gini importance 0.17. These were in similar regions to the most important features for Alpha and Rips based classifiers. Hence, we have reason to believe that regardless of the filtration method used, the ML model learned the importance of similar features which are distinguishing between the shapes.

    \begin{figure}[ht]
        \centering
        \includegraphics[width=\columnwidth]{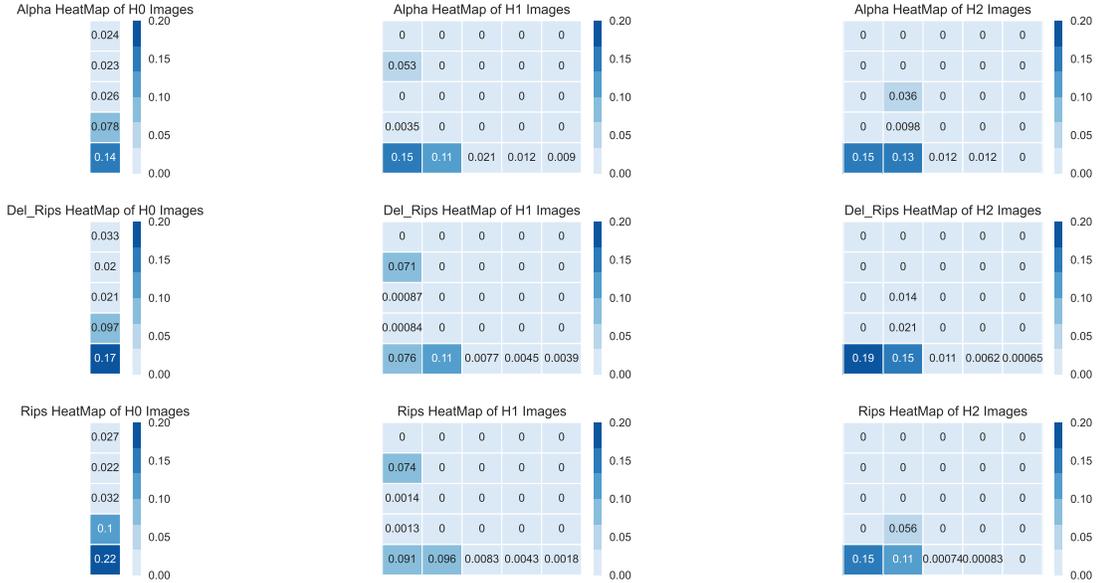}
        \caption{Heatmaps of $H_0$, $H_1$, and $H_2$ persistence diagram PI vectorizations. We took the 55 features for training the random forest classifier based on each filtration function and found their importance in their corresponding classifier models. The number in each pixel is the Gini importance of that pixel in our random forest's classification decisions. The higher the number, the more important the pixel.}
        \label{fig:heatmaps}
    \end{figure}

Our data analysis code is well-documented for reproducibility in our online repository \cite{synthetic_shape_classification}. 

\subsection{Classification of Sleep State}
\label{sec:ml-comparisons-real}
    We next investigate the applicability of DR to a classification problem involving biophysical data. The goal is to showcase the comparable effectiveness of DR with Alpha and Rips filtrations in terms of model performance metrics on a problem in which computational efficiency may be a relevant constraint. Our application is a reanalysis of the data and methodology employed in \cite{Chung_frontiers_hr_ml}, in which the authors develop an ML model to classify sleep stages of a participant using observed instantaneous heart rate (IHR) time series collected using EKG. A high-level visualization of the data-to-model development pipeline we implement is provided in Figure \ref{fig:ml_real_data_pipeline}. As a note, the authors of this manuscript (that investigates the properties and usage of the DR filtration) did not work directly with the data; instead, the processed data was obtained with permission from the authors of \cite{Chung_frontiers_hr_ml}.

    \begin{figure}[ht]
        \centering
        \includegraphics[width=\columnwidth]{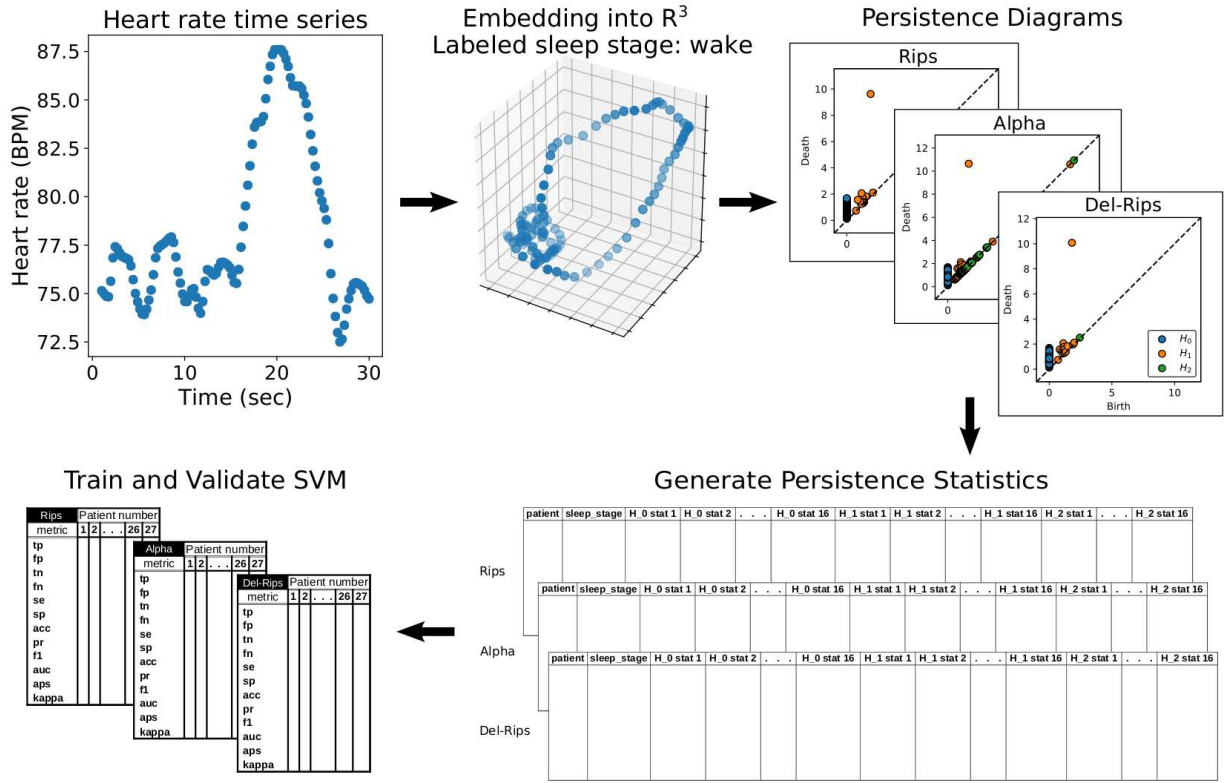}
        \caption{A flowchart of the data-to-model pipeline used. Starting in the upper-left, an example of a 30 second wake epoch followed by its delay embedding into $\R^3$, the corresponding persistence diagrams computed from DR, Alpha, and Rips filtrations on the cloud, and finally the subsequent vectorizations using persistence statistics which are used to train and validate a support vector machine binary classification model.}
        \label{fig:ml_real_data_pipeline}
    \end{figure}
    
    The authors of \cite{Chung_frontiers_hr_ml} trained a machine learning model to predict sleep state, trained on statistics of topological features of high dimensional point clouds that were built from delay embeddings of IHR time series. This approach is premised on the idea that delay embeddings of observed time series can recover dynamical features of underlying attractors in an unknown state space as implied by the well known theorem of Takens \cite{Takens}. The topological/geometric features of these attractors (as represented by PDs) may be discriminating between dynamic states. Applications based on similar rationale have also been found in contexts such as 3D motion capture data \cite{3d_motion_capture}. 

    The data on which models were trained and tested comprised 90 ECG recordings from Chang Gung Memorial Hospital (CGMH) as the CGMH-training database. The paper that originally used this data is \cite{malik2018sleep} and it provides more description into the data and data acquisition.
    
    Each recording was sampled at 200 Hz and each 30s epoch was manually annotated into either ``awake'' or ``asleep'' states. Sleep states were further categorized into one of 5 different sleep stages: stage 1 through 4 or REM sleep, giving a potential for a 6-class classification problem. To each recording a standard R-peak detection algorithm with a 5-beat median filter to remove artifacts was applied (similarly done in \cite{Xu2001AutomaticDO}). From this preprocessed data, IHR time series were computed at 4 Hz. For efficiency, we performed delay embeddings of the IHR time series into $\R^3$, choosing a delay parameter between consecutive coordinates of 5 (i.e., our delay vectors consisted of points $(x_{i}, x_{i+5}, x_{i+10}) \in \R^3$, taking every 5th point in the time series as the coordinate for our embedded points.) The number of indices between the starting coordinate of consecutive delay vectors (i.e., the stride) was set to 1 so as to include all timepoints in at least one delay vector. 
    
    After embedding each epoch into $\R^3$, we computed $H_0, H_1$, and $H_2$ persistence diagrams using the DR, Rips, and Alpha filtrations of the embedded point clouds. Following \cite{Chung_frontiers_hr_ml}, the resulting PDs were converted to fixed-length feature vectors by computing sample statistics of the persistence pairs in the diagrams. In particular, for each homological dimension $p=0,1,2$, we construct two sets: $M_p$, the set of means of all persistence pairs and $L_p$, the persistence of each pair (the death minus the birth). For each set, we calculate the mean, standard deviation, skew, kurtosis, 25th percentile, 50th percentile, 75th percentile, and persistent entropy \cite{chintakunta2015entropy}. The result was 16 persistence statistics corresponding to each dimension-$p$ persistence diagram. NaN was assigned to empty $H_p$ diagrams and these samples were later dropped from the training set. In total, from each sample a total of 48 persistence statistics (16 for each dimension $p=0,1,2$) were computed, resulting in a 48 dimensional feature space on which an SVM model was trained.

    A total of 67188 48-dimensional topological feature vectors (one for each 30 second epoch) were used to train a support vector machine with a linear kernel and balanced class weights to account for imbalance in the numbers of sleep and wake samples. In total 3 models were trained, one for each filtration method. The trained models were each validated on 27 held-out participants in the CGMH-validation by computing sensitivity (se), specificity (sp), accuracy (acc), precision (pr), F1 score (f1), AUC-score (auc), average precision score (aps), and kappa coefficient (kappa) for the epochs of each validation participant. To assess the statistical significance of any differences in performance metrics between DR and either Rips or Alpha filtrations, we applied Mood's median test from SciPy \cite{2020SciPy-NMeth} to the distributions of each performance metric across the 27 validation participants, comparing separately Rips to DR and Alpha to DR filtrations. In this context, Mood's median test tests the null hypothesis that the median of each performance metric is the same.
    
Our model construction differs from the one in the original paper in several key ways.  First, the authors of \cite{Chung_frontiers_hr_ml} derived topological features from considerably high dimensional (120) delay embeddings of the IHR time series. Such high dimensional embeddings pose a challenge for Alpha and DR filtrations since they both require computing Delaunay triangulations, which do not scale efficiently to high dimensions. Furthermore, the original model was trained on statistics derived only from $H_0$ and $H_1$ diagrams, while we included $H_2$ diagrams as well. We note that these changes come at a modest reduction in model performance compared to what is reported in \cite{Chung_frontiers_hr_ml}, although our goal was to assess differences between filtrations and not to improve over previously published classification models.

    \begin{table}
    \begin{center}
    \begin{tabular}{l|cc|cc|cc|}
     (a) & \multicolumn{2}{|c|}{Rips} & \multicolumn{2}{|c|}{Alpha} &
     \multicolumn{2}{|c|}{Del-Rips} \\
     \hline
      & median & iqr & median & iqr & median & iqr  \\
    se    &  0.565217 &  0.428794   &  0.569892 &  0.380411   &  0.555556 &  0.433198 \\
    sp    &  0.925065 &  0.271714   &  0.917511 &  0.270439   &  0.917271 &  0.272863 \\
    acc   &  0.852941 &  0.161052   &  0.848276 &  0.154095   &  0.852573 &  0.161564 \\
    pr    &  0.439636 &  0.316970   &  0.448718 &  0.344648   &  0.424307 &  0.301726 \\
    f1    &  0.517986 &  0.140499   &  0.489655 &  0.136024   &  0.472727 &  0.170183 \\
    auc   &  0.866199 &  0.083241   &  0.864330 &  0.079358   &  0.866646 &  0.096021 \\
    aps   &  0.594109 &  0.180105   &  0.555916 &  0.160430   &  0.552859 &  0.159507 \\
    kappa &  0.393895 &  0.178491   &  0.368380 &  0.197346   &  0.356587 &  0.192436 \\
    \end{tabular}
    \end{center}
    \begin{center}
    \begin{tabular}{l|c|c|}
     (b) & p-value for Rips vs Del-Rips & p-value for Alpha vs Del-Rips \\
     \hline
    se    &  1.000000 &   1.0 \\
    sp    &  1.000000 &   1.0 \\
    acc   &  1.000000 &   1.0 \\
    pr    &  1.000000 &   1.0 \\
    f1    &  0.276303 &   1.0 \\
    auc   &  1.000000 &   1.0 \\
    aps   &  0.102470 &   1.0 \\
    kappa &  1.000000 &   1.0 \\
    \end{tabular}
    \caption{\label{tab:results} (a) The median and interquartile ranges across 27 validation participants of each classificaiton model performance metric obtained from 3 classifiers trained either on topological features determine by DR, Rips, or Alpha filtrations. (b) The Mood's median test p-values after comparing DR model performance metrics against Rips and Alpha model performance metrics.}
    \end{center}
    \end{table}
    
    The median and interquartile ranges of each model performance metric across the validation participants as well as the p-values of Mood's median test of are shown in Table \ref{tab:results}. We observe very similar performance metrics for all three filtration methods, with either DR, Alpha, or Rips exhibiting the best median performance, depending on the metric. For Mood's median test, a small p-value indicates that there is a notable difference between the medians of the corresponding performance metric between the two classification models. Notice that the p-values in Table \ref{tab:results} are all very large, with the exception of the average precision score (aps) when comparing Rips vs DR. This suggests that, given the spread of model performance on different validation participants, we do not have enough evidence to reject the null hypothesis that the median performance metrics are the same.

Our data analysis code is well-documented for reproducibility in our online repository \cite{Amish_Mishra_ML-Del-Rips-sleep-wake-classification}.

\section{Conclusion}
    In this paper, we have defined and implemented the Delaunary-Rips (DR) filtration for point cloud data, compared its computational efficiency in practice to other standard filtrations built on point cloud data, characterized some of its stability properties, and demonstrated its application on various datasets in an ML pipeline. We proved that under sufficiently small perturbations of the locations of points in the cloud, the persistence diagrams vary continuously with the data, provided the underlying simplicial complex on the data set remains fixed. This was done by bounding from above the bottleneck distance between the persistence diagrams by twice the Hausdorff distance between the original data and its perturbed counterpart. In the case when the Delaunay triangulation changes as a result of a larger perturbation, we examined the instability of the DR filtration by carefully proving a discontinuity of the map between the data metric space and the diagram metric space. As far as we know, this result is a first of its kind to use the standard reduction algorithm to compute persistence diagrams on symbolic variables in service of a formal proof. Since the theoretical condition on stability may not always be met in practical applications we investigated whether the instability in the diagrams generated using DR poses a significant problem when used in an ML pipeline and found that it need not. 
    
    There are several limitations of this study and avenues of further inquiry. For one, we expect the relative performance of each filtration method used to derive topological features for an ML modelling task to be problem specific. Thus we cannot be certain the insensitivity of model performance to filtration method, and the comparable performance of DR to Rips and Alpha we observed will generalize to other contexts. 
    
    Our stability results were simplified by insisting we maintain the underlying Delaunay triangulation to keep the space of simplices fixed. While our results in section \ref{instability_section} indicate that, in general, a change in the underlying Delaunay triangulation can cause a discontinuity in the transformation sending data to diagram, a more precise characterization of the degree of the discontinuities is not known, and so we are limited to an empirical evaluation.
    
    In section \ref{sec:runtimes}, we compared a Python implementation of DR with other filtration methods. Implementation details including the choice of language may have implications for the relative performance gains of DR over other filtration methods. With a C/C++ implementation of DR, how does the runtime of computing persistence diagrams based on DR compare with runtimes of computing persistence diagrams using other complexes (Cech, Witness, etc.)? Along the same lines, how does our implementation (or a C/C++ implementation) of DR compare in runtime with implementations of other TDA methods in other software packages (e.g., Javaplex \cite{Javaplex}, Perseus \cite{perseus}, Dionysus \cite{dionysus}, Dipha \cite{dipha}, Gudhi \cite{gudhi})? Finally, during the drafting of this paper, an updated version of the software package Ripser was released named Ripser++ \cite{ripserplusplus}. The new approach makes use of GPU acceleration to parallelize processing simplices and finding persistence pairs. How does the runtime of DR for persistent diagram computation compare with that of Ripser++? Is there a way to harness parallelization to computationally benefit DR? Such systematic comparisons are outside the scope of this work.

\printbibliography[
heading=bibintoc,
title={Whole bibliography}
] 

@article{Roadmap,
  author =       {Nina Otter and Mason A Porter and Ulrike Tillmann and Peter Grindrod and Heather A Harrington},
  title =        "A roadmap for the computation of persistent homology",
  journal =      "EPJ Data Science",
  volume =       "6",
  number =       "17",
  year =         "2017",
  DOI =          "10.1140/epjds/s13688-017-0109-5",
}

@book{Edelsbrunner,
author = {Edelsbrunner, Herbert and Harer, John},
year = {2010},
month = {01},
pages = {},
title = {Computational Topology: An Introduction},
isbn = {978-0-8218-4925-5},
doi = {10.1007/978-3-540-33259-6_7}
}

@article{Bauer_2016,
   title={The Morse theory of Čech and Delaunay complexes},
   volume={369},
   ISSN={1088-6850},
   url={http://dx.doi.org/10.1090/tran/6991},
   DOI={10.1090/tran/6991},
   number={5},
   journal={Transactions of the American Mathematical Society},
   publisher={American Mathematical Society (AMS)},
   author={Bauer, Ulrich and Edelsbrunner, Herbert},
   year={2016},
   month={12},
   pages={3741–3762}
}

@article{Chung_frontiers_hr_ml,
  
AUTHOR={Chung, Yu-Min and Hu, Chuan-Shen and Lo, Yu-Lun and Wu, Hau-Tieng},   
	 
TITLE={A Persistent Homology Approach to Heart Rate Variability Analysis With an Application to Sleep-Wake Classification},      
	
JOURNAL={Frontiers in Physiology},      
	
VOLUME={12},      
	
YEAR={2021},      
	  
URL={https://www.frontiersin.org/article/10.3389/fphys.2021.637684},       
	
DOI={10.3389/fphys.2021.637684},      
	
ISSN={1664-042X}
}

@article{scikit-learn,
 title={Scikit-learn: Machine Learning in {P}ython},
 author={Pedregosa, F. and Varoquaux, G. and Gramfort, A. and Michel, V.
         and Thirion, B. and Grisel, O. and Blondel, M. and Prettenhofer, P.
         and Weiss, R. and Dubourg, V. and Vanderplas, J. and Passos, A. and
         Cournapeau, D. and Brucher, M. and Perrot, M. and Duchesnay, E.},
 journal={Journal of Machine Learning Research},
 volume={12},
 pages={2825--2830},
 year={2011}
}

@inproceedings{Weller1997StabilityOV,
  title={Stability of voronoi neighborship under perturbations of the sites},
  author={Frank-Uwe Weller},
  booktitle={CCCG},
  year={1997}
}

@InProceedings{Takens,
author="Takens, Floris",
editor="Rand, David
and Young, Lai-Sang",
title="Detecting strange attractors in turbulence",
booktitle="Dynamical Systems and Turbulence, Warwick 1980",
year="1981",
publisher="Springer Berlin Heidelberg",
address="Berlin, Heidelberg",
pages="366--381",
isbn="978-3-540-38945-3"
}

@INPROCEEDINGS{3d_motion_capture,
  author={Venkataraman, Vinay and Ramamurthy, Karthikeyan Natesan and Turaga, Pavan},
  booktitle={2016 IEEE International Conference on Image Processing (ICIP)}, 
  title={Persistent homology of attractors for action recognition}, 
  year={2016},
  volume={},
  number={},
  pages={4150-4154},
  doi={10.1109/ICIP.2016.7533141}}

@misc{ripserplusplus,
Author = {Simon Zhang, Mengbai Xiao and Hao Wang},
Title = {GPU-Accelerated Computation of Vietoris-Rips Persistence Barcodes},
Year = {2020},
Eprint = {arXiv:2003.07989},
}

@article{chintakunta2015entropy,
  title={An entropy-based persistence barcode},
  author={Chintakunta, Harish and Gentimis, Thanos and Gonzalez-Diaz, Rocio and Jimenez, Maria-Jose and Krim, Hamid},
  journal={Pattern Recognition},
  volume={48},
  number={2},
  pages={391--401},
  year={2015},
  publisher={Elsevier}
}

@book{hatcher2002algebraic,
  title={Algebraic Topology},
  author={Hatcher, A.},
  isbn={9780521795401},
  lccn={00065166},
  series={Algebraic Topology},
  url={https://books.google.com/books?id=BjKs86kosqgC},
  year={2002},
  publisher={Cambridge University Press}
}

@article{leibon2008topological,
  title={Topological structures in the equities market network},
  author={Leibon, Gregory and Pauls, Scott and Rockmore, Daniel and Savell, Robert},
  journal={Proceedings of the National Academy of Sciences},
  volume={105},
  number={52},
  pages={20589--20594},
  year={2008},
  publisher={National Acad Sciences}
}

@inproceedings{chung2009persistence,
  title={Persistence diagrams of cortical surface data},
  author={Chung, Moo K and Bubenik, Peter and Kim, Peter T},
  booktitle={International Conference on Information Processing in Medical Imaging},
  pages={386--397},
  year={2009},
  organization={Springer}
}

@article{QAISER2016119,
title = {Persistent Homology for Fast Tumor Segmentation in Whole Slide Histology Images},
journal = {Procedia Computer Science},
volume = {90},
pages = {119-124},
year = {2016},
note = {20th Conference on Medical Image Understanding and Analysis (MIUA 2016)},
issn = {1877-0509},
doi = {10.1016/j.procs.2016.07.033},
url = {https://www.sciencedirect.com/science/article/pii/S1877050916312133},
author = {Talha Qaiser and Korsuk Sirinukunwattana and Kazuaki Nakane and Yee-Wah Tsang and David Epstein and Nasir Rajpoot}
}

@article{kramar2014quantifying,
  title={Quantifying force networks in particulate systems},
  author={Kram{\'a}r, Miroslav and Goullet, Arnaud and Kondic, Lou and Mischaikow, Konstantin},
  journal={Physica D: Nonlinear Phenomena},
  volume={283},
  pages={37--55},
  year={2014},
  publisher={Elsevier}
}

@InProceedings{10.1007/978-3-319-64185-0_11,
author="Asaad, Aras
and Jassim, Sabah",
editor="Kraetzer, Christian
and Shi, Yun-Qing
and Dittmann, Jana
and Kim, Hyoung Joong",
title="Topological Data Analysis for Image Tampering Detection",
booktitle="Digital Forensics and Watermarking",
year="2017",
publisher="Springer International Publishing",
address="Cham",
pages="136--146",
isbn="978-3-319-64185-0"
}

@article{bauer2021ripser,
  title={Ripser: efficient computation of Vietoris--Rips persistence barcodes},
  author={Bauer, Ulrich},
  journal={Journal of Applied and Computational Topology},
  volume={5},
  number={3},
  pages={391--423},
  year={2021},
  publisher={Springer}
}

@article{Čufar2020, doi = {10.21105/joss.02614}, url = {https://doi.org/10.21105/joss.02614}, year = {2020}, publisher = {The Open Journal}, volume = {5}, number = {54}, pages = {2614}, author = {Matija Čufar}, title = {Ripserer.jl: flexible and efficient persistent homology computation in Julia}, journal = {Journal of Open Source Software} }

@article{Edelsbrunner1993TheUO,
  title={The union of balls and its dual shape},
  author={Herbert Edelsbrunner},
  journal={Discrete \& Computational Geometry},
  year={1993},
  volume={13},
  pages={415-440}
}

@inproceedings{sheehy2012linear,
  title={Linear-size approximations to the Vietoris-Rips filtration},
  author={Sheehy, Donald R},
  booktitle={Proceedings of the twenty-eighth annual symposium on Computational geometry},
  pages={239--248},
  year={2012}
}

@article{Guibas2007ReconstructionUW,
  title={Reconstruction Using Witness Complexes},
  author={Leonidas J. Guibas and Steve Oudot},
  journal={Discrete \& Computational Geometry},
  year={2007},
  volume={40},
  pages={325-356}
}

@inproceedings{Silva2004TopologicalEU,
  title={Topological estimation using witness complexes},
  author={Vin de Silva and Gunnar E. Carlsson},
  booktitle={Symposium on Point Based Graphics},
  year={2004}
}

@article{de2008weak,
  title={A weak characterisation of the Delaunay triangulation},
  author={Vin De Silva},
  journal={Geometriae Dedicata},
  volume={135},
  number={1},
  pages={39--64},
  year={2008},
  publisher={Springer}
}

@article{ctralie2018ripser,
  doi = {10.21105/joss.00925},
  url = {https://doi.org/10.21105/joss.00925},
  year  = {2018},
  month = {9},
  publisher = {The Open Journal},
  volume = {3},
  number = {29},
  pages = {925},
  author = {Christopher Tralie and Nathaniel Saul and Rann Bar-On},
  title = {{Ripser.py}: A Lean Persistent Homology Library for Python},
  journal = {The Journal of Open Source Software}
}

@software{Amish_Mishra_Delaunay-Rips,
author = {Amish Mishra},
title = {{Delaunay-Rips}},
url = {https://github.com/amish-mishra/cechmate-DR}
}

@software{synthetic_shape_classification,
author = {Amish Mishra},
title = {Classification of synthetic data into shape classes},
url = {https://github.com/amish-mishra/TDA_shape_classification_using_DR}
}

@software{perseus,
author = {Vidit Nanda},
title = {Perseus, the Persistent Homology Software},
url = {http://www.sas.upenn.edu/~vnanda/perseus},
note = {Accessed 03/01/2023}}

@software{dionysus,
author = {Dmitriy Morozov},
title = {Dionysus},
url = {https://pypi.org/project/dionysus/}}

@software{dipha,
author = {Jan Reininghaus},
title = {DIPHA (A Distributed Persistent Homology Algorithm)},
url = {https://github.com/DIPHA/dipha}}

@inproceedings{gudhi,
author = {Maria, Clément and Boissonnat, Jean-Daniel and Glisse, Marc and Yvinec, Mariette},
year = {2014},
month = {06},
pages = {},
title = {The Gudhi Library: Simplicial Complexes and Persistent Homology},
isbn = {978-3-662-44198-5},
doi = {10.1007/978-3-662-44199-2_28}
}

@software{Amish_Mishra_ML-Del-Rips-sleep-wake-classification,
author = {Amish Mishra},
title = {{ML-Del-Rips-sleep-wake-classification}},
url = {https://github.com/amish-mishra/ML-Del-Rips-sleep-wake-classification}
}

@software{cechmate,
author = {Chris Tralie and Nathaniel Saul},
title = {{Cechmate}},
url = {https://github.com/scikit-tda/cechmate},
note    = {License: MIT}
}

@article{BAUER201776,
title = {Phat – Persistent Homology Algorithms Toolbox},
journal = {Journal of Symbolic Computation},
volume = {78},
pages = {76-90},
year = {2017},
note = {Algorithms and Software for Computational Topology},
issn = {0747-7171},
doi = {10.1016/j.jsc.2016.03.008},
url = {https://www.sciencedirect.com/science/article/pii/S0747717116300098},
author = {Ulrich Bauer and Michael Kerber and Jan Reininghaus and Hubert Wagner}
}

@software{persim,
author = {Nathaniel Saul and Chris Tralie and Francis Motta and Michael Catanzaro and Gabrielle Angeloro and Calder Sheagren},
title = {{Persim}},
url = {https://persim.scikit-tda.org/en/latest/},
note    = {License: MIT}
}

@article{ADAMASZEK20171,
title = {Random cyclic dynamical systems},
journal = {Advances in Applied Mathematics},
volume = {83},
pages = {1-23},
year = {2017},
issn = {0196-8858},
doi = {10.1016/j.aam.2016.08.007},
url = {https://www.sciencedirect.com/science/article/pii/S0196885816300768},
author = {Michał Adamaszek and Henry Adams and Francis Motta}
}

@incollection{gudhi:AlphaComplex
, author    = {Vincent Rouvreau}
, title     = {Alpha complex}
, publisher = {GUDHI Editorial Board}
, edition   = {3.7.1}
, booktitle = {GUDHI User and Reference Manual}
, url       = {https://gudhi.inria.fr/doc/3.7.1/group__alpha__complex.html}
, year      = {2023}
}

@article{Chazal_gromov_stabililty,
author = {Chazal, Frédéric and Cohen-Steiner, David and Guibas, Leonidas J. and Mémoli, Facundo and Oudot, Steve Y.},
title = {Gromov-Hausdorff Stable Signatures for Shapes using Persistence},
journal = {Computer Graphics Forum},
volume = {28},
number = {5},
pages = {1393-1403},
doi = {10.1111/j.1467-8659.2009.01516.x},
url = {https://onlinelibrary.wiley.com/doi/abs/10.1111/j.1467-8659.2009.01516.x},
eprint = {https://onlinelibrary.wiley.com/doi/pdf/10.1111/j.1467-8659.2009.01516.x},
year = {2009}
}

@article{CohenSteiner2005StabilityOP,
  title={Stability of Persistence Diagrams},
  author={David Cohen-Steiner and Herbert Edelsbrunner and John Harer},
  journal={Discrete \& Computational Geometry},
  year={2005},
  volume={37},
  pages={103-120}
}

@article{adams2017persistence,
  title={Persistence images: A stable vector representation of persistent homology},
  author={Adams, Henry and Emerson, Tegan and Kirby, Michael and Neville, Rachel and Peterson, Chris and Shipman, Patrick and Chepushtanova, Sofya and Hanson, Eric and Motta, Francis and Ziegelmeier, Lori},
  journal={Journal of Machine Learning Research},
  volume={18},
  year={2017}
}

@article{boissonnat2013stability,
  title={The stability of Delaunay triangulations},
  author={Boissonnat, Jean-Daniel and Dyer, Ramsay and Ghosh, Arijit},
  journal={International Journal of Computational Geometry \& Applications},
  volume={23},
  number={04n05},
  pages={303--333},
  year={2013},
  publisher={World Scientific}
}

@article{Skraba2020WassersteinSF,
  title={Wasserstein Stability for Persistence Diagrams},
  author={Primoz Skraba and Katharine Turner},
  journal={arXiv: Algebraic Topology},
  year={2020}
}

@article{lc-reduction,
author = {Matoušek, Jiří},
year = {2008},
month = {01},
pages = {},
title = {LC reductions yield isomorphic simplicial complexes},
volume = {3},
journal = {Contributions to Discrete Mathematics [electronic only]}
}

@inproceedings{ribeiro2016should,
  title={" Why should i trust you?" Explaining the predictions of any classifier},
  author={Ribeiro, Marco Tulio and Singh, Sameer and Guestrin, Carlos},
  booktitle={Proceedings of the 22nd ACM SIGKDD international conference on knowledge discovery and data mining},
  pages={1135--1144},
  year={2016}
}

@article{Krzywinski2017PointsOS,
  title={Points of Significance: Classification and regression trees},
  author={Martin Krzywinski and Naomi Altman},
  journal={Nature Methods},
  year={2017},
  volume={14},
  pages={757-758}
}

@article{malik2018sleep,
  title={Sleep-wake classification via quantifying heart rate variability by convolutional neural network},
  author={Malik, John and Lo, Yu-Lun and Wu, Hau-tieng},
  journal={Physiological measurement},
  volume={39},
  number={8},
  pages={085004},
  year={2018},
  publisher={IOP Publishing}
}

@article{Xu2001AutomaticDO,
  title={Automatic detection of artifacts in heart period data.},
  author={Xirong Xu and Stephanie Schuckers},
  journal={Journal of electrocardiology},
  year={2001},
  volume={34 Suppl},
  pages={
          205-10
        }
}

@ARTICLE{2020SciPy-NMeth,
  author  = {Virtanen, Pauli and Gommers, Ralf and Oliphant, Travis E. and
            Haberland, Matt and Reddy, Tyler and Cournapeau, David and
            Burovski, Evgeni and Peterson, Pearu and Weckesser, Warren and
            Bright, Jonathan and {van der Walt}, St{\'e}fan J. and
            Brett, Matthew and Wilson, Joshua and Millman, K. Jarrod and
            Mayorov, Nikolay and Nelson, Andrew R. J. and Jones, Eric and
            Kern, Robert and Larson, Eric and Carey, C J and
            Polat, {\.I}lhan and Feng, Yu and Moore, Eric W. and
            {VanderPlas}, Jake and Laxalde, Denis and Perktold, Josef and
            Cimrman, Robert and Henriksen, Ian and Quintero, E. A. and
            Harris, Charles R. and Archibald, Anne M. and
            Ribeiro, Ant{\^o}nio H. and Pedregosa, Fabian and
            {van Mulbregt}, Paul and {SciPy 1.0 Contributors}},
  title   = {{{SciPy} 1.0: Fundamental Algorithms for Scientific
            Computing in Python}},
  journal = {Nature Methods},
  year    = {2020},
  volume  = {17},
  pages   = {261--272},
  adsurl  = {https://rdcu.be/b08Wh},
  doi     = {10.1038/s41592-019-0686-2},
}

@inproceedings{Javaplex,
     author = {Tausz, Andrew and Vejdemo-Johansson, Mikael and Adams, Henry},
     title = {Java{P}lex: {A} research software package for persistent (co)homology},
     booktitle = {Proceedings of ICMS 2014},
     editor = {Hong, Han and Yap, Chee},
     series = {Lecture Notes in Computer Science 8592}, 
     year = {2014},
     pages = {129-136},
     note = {Software available at \url{http://appliedtopology.github.io/javaplex/}}
}


\end{document}